\documentclass[a4paper]{llncs}

\usepackage[normalem]{ulem}
\usepackage{amssymb,amsmath}
\usepackage[]{color}
\usepackage{comment}
\usepackage{paralist}
\usepackage[mathscr]{eucal}
\usepackage[capitalize]{cleveref}
\usepackage{algorithm}
\usepackage[noend]{algorithmic}
\usepackage{algorithmicext}
\usepackage{multirow}
\usepackage{tikz}
\usepackage{calc}
\usetikzlibrary{arrows,decorations.pathmorphing,backgrounds,positioning,fit,automata,shapes,calc}
\usepackage{exscale}
\usepackage{xspace}
\usepackage{setspace}
\usepackage{fancyhdr}
\usepackage{cite}
\usepackage{subfig}
\usepackage{caption}    %
\usepackage{layout}
\usepackage{multicol}
\usepackage{url}
\usepackage{wasysym}
\usepackage{etoolbox}
\newenvironment{pf}{\begin{proof}}{\end{proof}}

{\bfseries}{\itshape}

\newconstruct{\ON}{\textbf{on}}{\textbf{do}}{\ENDON}{\textbf{end on}}

\Crefname{assumption}{Asmp.}{Asmps.}
\Crefname{lemma}{Lem.}{Lems.}
\Crefname{definition}{Def.}{Defs.}
\Crefname{algorithm}{Alg.}{Algs.}
\Crefname{ALC@unique}{line}{lines}

\def\EMPTY{\vskip8pt minus 4pt\relax}
\gdef\dash---{\thinspace---\hskip.16667em\relax}

\def\true{\mbox{\sc true}}
\def\false{\mbox{\sc false}}
\def\Er{\ensuremath{E^r}}
\def\Gr{\ensuremath{\mathcal{G}^r}}
\def\Ga{\mathcal{G}}

\def\N{\mathcal{N}}
\def\U{\mathcal{U}}
\def\A{\mathcal{A}}
\def\B{\mathcal{B}}
\def\V{\mathcal{V}}

\def\stableSCC{\mbox{\tt InStableRoot}}
\def\getred{\lambdaabstraction}
\def\distx_#1{\mbox{\tt{cd}}^{#1}}
\makeatletter
\def\dist{\@ifnextchar_\distx{\distx_}}
\makeatother
\def\Knows#1{\mathsf{K}_p}
\def\Everyknows{\mathsf{E}}

\newconstruct{\PROC}{\textbf{procedure}}{}{\ENDPROC}{\textbf{end on}}
\newconstruct{\PRED}{\textbf{predicate}}{}{\ENDPRED}{\textbf{end on}}
\newconstruct{\FUNC}{\textbf{function}}{}{\ENDFUNC}{\textbf{end on}}

\renewcommand{\leq}{\leqslant}
\renewcommand{\le}{\leqslant}
\renewcommand{\geq}{\geqslant}
\renewcommand{\ge}{\geqslant}

\gdef\dash---{\thinspace---\hskip.16667em\relax}

\gdef\op|{\,|\;}
\gdef\op:{\,:\;}
\newcommand{\li}[1]{\langle#1\rangle}

\def\la{\leftarrow}

\def\ra{\rightarrow}
\def\lt{\leadsto}

\def\set#1{\left\{#1\right\}}

\def\G{\mathcal{G}}

\def\D{\mathcal{D}}
\def\S{\mathscr{S}}
\def\L{\mathcal{L}}
\def\P{\mathcal{P}}
\def\Q{\mathcal{Q}}

\newcommand{\M}{\mathcal{M}}
\newcommand{\R}{\mathcal{R}}

\newcommand{\msg}[1]{\langle#1\rangle}
\newcommand{\APiD}{{A_{|\PiD}}}
\newcommand{\ARes}{{A_{|\D}}}

\newcommand{\PiD}{\overline{\D}}
\newcommand{\Rii}{\R_\text{($\PiD$)}}
\newcommand{\Riii}{\R_\text{($\D,\PiD$)}}

\newcommand{\Mod}{\mathscr{M}}

\newcommand{\condNonempty}{(A)}
\newcommand{\condRiiRiii}{(B)}
\newcommand{\condNocons}{(C)}
\newcommand{\condMADMA}{(D)}

\makeatletter
\def\leftNumbered{\tagsleft@true}\def\rightNumbered{\tagsleft@false}
\makeatother

\def\Timely{\mathcal{N}}

\newcommand{\kommentar}[1]{\noindent\textbf{KOMMENTAR: }\marginpar{****}%
\textit{#1}\textbf{ :RATNEMMOK}}

\def\Er{\ensuremath{E^r}}
\def\Gr{\ensuremath{\mathcal{G}^r}}
\def\Gcap(#1){\mathcal{G}^{\cap\, #1}}
\def\Ecap(#1){E^{\cap\, #1}}

\def\N{\mathcal{N}}

\newcommand{\edge}[1]{\stackrel{#1}{\ra}}
\newcommand{\ltedge}[1]{\stackrel{#1}{\lt}}

\def\R{\mathcal{R}}

\newcommand{\indist}[1]{\ensuremath{\overset{#1}{\sim}}}
\newcommand{\ksa}{$k$-set agreement}
\newcommand{\compat}[1]{\ensuremath{\preccurlyeq_{#1}}}
\newcommand{\mH}{\ensuremath{\mathcal{H}}}
\newcommand{\rGST}{\ensuremath{r_{ST}}}

\newcommand{\subruns}[1][\PiD]{\preccurlyeq_{#1}}
\def\MAD{\M_{A_{|\PiD}}}

\newcommand{\majinfsymb}{{\hookrightarrow}_{\tt m}}
\newcommand{\majinf}[2]{{#1} \majinfsymb {#2}}
\newcommand{\stronginf}[2]{{#1} \majinfsymb* {#2}}

\newcommand{\influence}[2]{{#1} {\hookrightarrow} {#2}}
\newcommand{\goodD}{$D$-bounded}
\newcommand{\goodDE}{$\nwbound$-network-bounded}
\newcommand{\almostgoodDE}{almost $\nwbound-1$-network-bounded}

\newcommand{\lambdaabstraction}{\mathtt{GetLock}}
\def\cdiam{\diameter}
\def\cheight{h}
\def\hist{\ensuremath{\mbox{\tt hist}}}
\def\nwbound{H}

\newcommand{\MAd}[1]{\ensuremath{\mbox{VSRC}({#1})}}

\newcommand{\MAdv}[2]{\ensuremath{\mbox{VSRC}({#1},{#2})}}
\newcommand{\MAdvI}[2]{\ensuremath{\mbox{VSRC}({n},{#2})+\mbox{MAJINF}({#1})}}

\newcommand{\MAkd}[1]{\MAdv{k}{{#1}}}
\newcommand{\MAkI}{\MAdvI{k}{3D+\nwbound}}
\newcommand{\MAkdI}[1]{\MAdvI{k}{{#1}}}
\newcommand{\MAdp}[1]{\ensuremath{\mbox{VSRC'}({#1})}}
\newcommand{\MASigma}{\ensuremath{\mbox{$\Sigma_{n-1}$-MAJ}}}
\newcommand{\MAJINF}[1]{\ensuremath{\mbox{MAJINF}({#1})}}

\newtoggle{journal}
 \toggletrue{journal}

\iftoggle{journal}{ 
\usepackage{fullpage}
}{}

\begin{document}

\title{Gracefully Degrading Consensus and $k$-Set Agreement in Directed Dynamic Networks}

\author{Martin Biely\inst{1} \and Peter Robinson\inst{2} \and Ulrich Schmid\inst{3} \and Manfred Schwarz\inst{3} \and Kyrill Winkler\inst{3}}
\institute{EPFL, Switzerland, \email{martin.biely@epfl.ch} \and National University of Singapore, Singapore, \email{robinson@comp.nus.edu.sg} \and ECS Group, TU Wien, Austria, \email{\{s,mschwarz,kwinkler\}@ecs.tuwien.ac.at}}

\maketitle

\begin{abstract}
We study distributed agreement in synchronous directed
dynamic networks, where an omniscient message adversary controls the 
presence/absence of communication links.
We prove that consensus is impossible under a message adversary that
guarantees weak connectivity only, and introduce vertex-stable root 
components (VSRCs) as a means for circumventing this impossibility: A 
\MAkd{d} message adversary guarantees that, eventually,
there is an interval of $d$ consecutive rounds where every 
communication graph contains at most $k$ strongly connected components
consisting of the same processes (with possibly varying interconnect
topology), which have at most out-going links to the 
remaining processes. We present a consensus algorithm that works correctly under a \MAdv{1}{4H+2}
message adversary, where $H$ is the dynamic causal network diameter.
Our algorithm maintains local estimates of the communication graphs, 
and applies techniques for detecting network stability and 
univalent system configurations. Several related impossibility results and lower 
bounds, in particular, that neither a \MAdv{1}{H-1} message adversary 
nor a \MAdv{2}{\infty} one allow to solve consensus, reveal that there
is not much hope to deal with (much) stronger message adversaries here.

However, we show that gracefully degrading consensus, which
degrades to general $k$-set agreement in case of unfavorable network
conditions, allows to cope with stronger message adversaries:
We provide a $k$-uniform $k$-set agreement algorithm, where the number
of system-wide decision values $k$ is not encoded in the algorithm, but
rather determined by the actual power of the message adversary in a run:
Our algorithm guarantees at most $k$ decision values under a \MAkdI{d} 
message adversary, which combines \MAdv{n}{d} 
(with some small value of $d$, ensuring termination)
with some information flow guarantee \MAJINF{k} between certain VSRCs (ensuring
$k$-agreement). Since related
impossibility results reveal that a \MAkd{d} message adversary
is too strong for solving $k$-set agreement and that some information
flow between VSRCs is mandatory for this purpose as well, our results
provide a significant step towards the exact solvability/impossibility
border of general $k$-set agreement in directed dynamic networks.

\end{abstract}

\section{Introduction}
\pagestyle{plain}

\label{sec:intro}

Dynamic networks, instantiated, e.g., by 
wireless sensor networks, mobile ad-hoc networks and vehicle
area networks, are becoming
ubiquitous nowadays. The primary properties of such networks are sets 
of participants (called processes in the sequel) that are a priori unknown and potentially 
changing, time-varying connectivity between processes, and 
the absence of a central control. Dynamic networks is an important and
very active area of research \cite{KO11:SIGACT}.

Accurately modeling dynamic networks is challenging, for several reasons:
First, process mobility, process crashes/recoveries, 
deliberate joins/leaves, and peculiarities in the low-level 
system design like duty-cycling (used to save energy 
in wireless sensor networks) make static communication 
topologies, as typically used in classic 
network models, inadequate for dynamic networks.
Certain instances of dynamic networks, in particular, peer-to-peer networks 
\cite{KSW10} and inter-vehicle area networks \cite{FNSA12}, even
suffer from
significant churn, i.e., a large number of processes that can appear/disappear
over time, possibly in the presence of faulty 
processes \cite{APR13}, and hence consist of a potentially unbounded total number of
participants over time. More classic applications
like \emph{mobile ad-hoc networks} (MANETS) 
\cite{KM07}, wireless sensor networks \cite{ASSC02,YMG08}
and disaster relief applications \cite{LHSP11} typically consist of
a \emph{bounded} (but typically unknown) total number of processes.

Second, communication in many dynamic networks, in particular,
in wireless networks like MANETS, is inherently broadcast: When a process 
transmits, then every other process within its transmission range will
observe this transmission --- either by legitimately receiving the
message or as some form of interference. This creates quite irregular
communication behavior, such as capture effects and
near-far problems \cite{WJC00}, where certain (nearby) transmitters may
``lock'' a receiver and thus prohibit the reception of messages
from other senders. Consequently, the ``health'' of a wireless
link between two processes may vary heavily over time \cite{CWPE05}.
For low-bandwidth wireless transceivers, an acceptable link quality
usually even requires communication scheduling \cite{ROG03}
(e.g., time-slotted communication) 
for reducing the mutual interference. Overall, this 
results in a frequently changing spatial distribution of pairs of nodes 
that can communicate at a given point in time. 

As a consequence, many dynamic networks, in particular, wireless ones \cite{CWKP05}, 
are not adequately modeled by means of bidirectional links: Fading and interference phenomenons
\cite{SBB12,GKFS10:perf}, including capture effects and near-far problems,
are \emph{local} effects
that affect only the receiver of a wireless link. 
Given that the sender, which is also the receiver of the reverse link, resides
at a different location, the two receivers are likely to
experience very different levels of fading and interference \cite{FWZ05}.
This effect is even more pronounced in the case of time-slotted communication,
where forward and backward links are used at different times. Consequently,
the existence of asymmetric communication links cannot be ruled out in 
practice: According to \cite{NKYG07}, 80\% of the links in a typical
wireless network are asymmetric.

Despite these facts, most of the dynamic network research we are 
aware of assumes bidirectional links \cite{KLO10:STOC,KOM11}.
The obvious advantage of this abstraction is simplicity of the
algorithm design, as strong communication guarantees obviously
make this task easier. Moreover, it allows the re-use of existing 
techniques for wireline networks, which naturally support
bidirectional communication.
However, there are also major disadvantages of this convenient
abstraction:
First, for dynamic networks that operate in environments with unfavourable
communication conditions, e.g. in disaster relief applications or, more 
generally, in settings with various interferers and obstacles 
that severely inhibit communication, bidirectional links may simply 
not be achievable. For implementing distributed services in such settings, 
algorithms that do not need bidirectional links are mandatory.
Second, the entire system needs to be engineered in such a way that bidirectional 
single-hop communication can be provided within bounded time.
This typically requires relatively dense networks and/or 
processes that are equipped with powerful
communication interfaces, which incur significant cost when compared
to sparser networks or/and cheaper or more energy-saving communication
devices. 
And last but not least, if directed single-hop communication was already 
sufficient to reach some desired goal (say, reaching some destination
process) via multi-hop messages, waiting for guaranteed single-hop
bidirectional communication would incur a potentially significant,
unnecessary delay. Obviously, in such settings, algorithmic 
solutions that do not need bidirectional single-hop communication 
could be significantly faster.

\medskip

In this paper, we thus restrict our attention to dynamic networks
consisting of an \emph{unknown but bounded} total number of processes, 
which are interconnected by
\emph{directed} communication links. The system is assumed to
be synchronous,\footnote{As synchronized clocks 
are typically required for basic communication 
in wireless systems anyway, e.g., for transmission scheduling and 
sender/receiver synchronization, this is not an unrealistic 
assumption: Global synchrony can be implemented
directly at low system levels, e.g., via IEEE~1588 network time
synchronization or GPS receivers, or at higher levels via  time synchronization protocols like FTSP \cite{MKSL04} or even synchronizers \cite{Awe85}.} 
hence time is measured in discrete \emph{rounds} that allow
the processes to exchange at most one message. Time-varying communication 
is modeled as a sequence of \emph{communication graphs}, which contain
a directed edge between two processes if the message sent in the
corresponding round is successfully received. A bidirectional 
link is modeled by a pair of directed links that are considered 
independent of each other here. 

A natural approach to build robust services despite the dynamic
nature of such systems is to use some sort of distributed agreement
on certain system parameters like schedules, frequencies, and operating modes,
as well as on application-level issues: Such a solution allows to use arbitrary
algorithms for generating local proposals, which are supplied as inputs to a 
consensus algorithm that finally selects one of them consistently at all 
processes. As opposed to master-slave-based solutions, this
approach avoids the single point of failure formed by the process acting 
as the master.

The ability to reach \emph{system-wide} consensus is hence the most
convenient abstraction one could provide here. The first\footnote{A preliminary version of this part of our paper 
has appeared at SIROCCO'12 \cite{BRS12:sirocco}.} major contribution
of our paper 
is hence a suite of impossibility results and a consensus algorithm
for directed dynamic networks that, to the best of our knowledge,
works under the weakest communication guarantees sufficient for consensus
known so far.

Obviously, however, one cannot reasonably assume that every dynamic network 
always provides sufficiently strong communication guarantees for solving 
consensus. Fortunately, weaker forms of distributed agreement
are sufficient for certain applications. In case of determining communication
schedules \cite{ROG03}, for example, which are used for staggering message transmission of
nearby nodes in time to decrease mutual interference, it usually suffices if
those processes that have to communicate regularly with each other (e.g.,
for implementing a distributed service within a partition) 
agree on their schedule. A more high-level example would be
agreement on rescue team membership \cite{GZCW10} in disaster relief 
applications.

For such applications, suitably designed $k$-set agreement algorithms \cite{Cha93}, 
where processes must agree on at most $k$ different values system-wide, 
are a viable alternative to consensus ($k=1$). This is particularly
true if such a \ksa{} (i) respects partitions, in the sense that processes 
in the same (single) partition decide on the same value, and (ii) is
\emph{gracefully degrading}, in the sense that the actual number $k$ of different decision
values depends on the \emph{actual} network topology in the 
execution: If the network is well-behaved, the resulting $k$ is
small (ideally, $k=1$), whereas $k$ may increase under unfavorable
conditions. Whereas any gracefully degrading algorithm must be 
$k$-uniform, i.e., unaware of any a priori information on $k$, 
it should ideally also be \emph{$k$-optimal}, i.e., 
produce the smallest number $k$ of different decisions possible.

The second\footnote{A brief announcement of this part of our paper appeared at PODC'14 \cite{SWSBR13:PODC}.} major contribution of our paper are several
impossibility results
for $k$-set agreement in directed dynamic networks, as well as the, to the
best of our knowledge, first instance of a worst-case $k$-optimal \ksa{}, 
i.e., a consensus algorithm that indeed degrades \emph{gracefully} to general 
$k$-set agreement.

\medskip
\noindent
{\bf Detailed contributions and paper organization.} 

In \cref{sec:model}, we introduce our detailed system model, which adopts 
the \emph{message adversary} notation used in \cite{RS13:PODC}. 
It consists of an (unknown) number $n$ of processes, where communcation is
modeled by a sequence of directed communication graphs, one for each
round: If some edge $(p,q)$ is present in the communication graph $\Gr$ of round $r$, 
then process $q$ has received the message sent to it by $p$ in round $r$. 
The message
adversary determines the set of links actually present in every $\Gr$, 
according to certain constraints that may be viewed as network assumptions.

With respect to consensus, we provide the following contributions:
\begin{enumerate}
\item[(1)] In \cref{sec:consimposs}, we show that communication graphs that are weakly connected in
every round are not sufficient for solving consensus, and introduce
a fairly weak additional assumption that allows to overcome this
impossibility. Our message adversary \MAd{d} requires that the communication
graph in every round is weakly connected
and has one (possibly changing) strongly connected component (called a 
\emph{root component}) that has no in-coming links from processes
outside. Note carefully that every directed graph 
has at least one root component. Since this assumption is still
too weak for solving consensus, \MAd{d} also requires that, eventually, there
will be $d$ consecutive rounds where the processes in the
root component remain the same, although the connection topology
may still change. We use the term \emph{vertex-stable root component} (VSRC)
for this requirement. In \cref{sec:consensus},
we provide a consensus algorithm that works in this model, and
prove its correctness. Our algorithm requires a window of stability
of $d=4\nwbound+2$ rounds, where $\nwbound$ is the \emph{dynamic network causal diameter} 
of the network
(= the number of rounds required to reach all processes in the network 
from any process in the vertex-stable root component via multi-hop communication).

\item[(2)] In \cref{sec:consimposs}, we show that any consensus and leader election algorithm has to
know an a priori bound on $\nwbound$. Since $n-1$ is a
trivial bound on $\nwbound$, this implies that no uniform algorithm, i.e.,
no algorithm unaware of $n$ or $\nwbound$, 
can solve consensus in our model.
In addition, we prove that consensus is impossible both under \MAdv{2}{\infty} and
under \MAd{\nwbound-1}, which shows that $\nwbound$ is a lower bound for the window of stability of VSRCs.
We also demonstrate that neither reliable broadcast, atomic broadcast, nor 
  causal-order broadcast can be implemented under \MAd{d}. The same is shown to
be true for counting, $k$-verification, $k$-token dissemination, all-to-all token 
  dissemination, and $k$-committee election.
\end{enumerate}

With respect to $k$-set agreement and gracefully degrading consensus, we
provide the following contributions:
\begin{enumerate}
\item[(3)] In \cref{sec:impossibility-proofs}, we 
provide a fairly weak natural message adversary \MAkd{d} that
is still too strong for solving $k$-set agreement:
It reveals that the restriction to at most $k$
simultaneous VSRCs in every round is \emph{not} sufficient for solving $k$-set agreement if just a single VSRC is vertex-stable for less than $n-k$ rounds: A
generic reduction of $k$-set agreement to consensus introduced in
\cite{BRS11:OPODIS}, in conjunction with certain bivalence arguments, is used
to construct a non-terminating run in this case.
Moreover, \emph{eventual} stability of \emph{all} VSRCs 
is also not enough for solving \ksa{}, not even when it is 
guaranteed that (substantially) less than $k$ VSRCs exist simultaneously.
The latter is a consequence of some adversarial
partitioning over time, which could happen in our dynamic networks.

\item[(4)] In \cref{sec:sufficiency}, we show that the message adversary 
\MAkdI{d}, which combines \MAdv{n}{d} (ensuring termination)
with some information flow guarantee \MAJINF{k} between certain VSRCs 
(ensuring $k$-agreement), is sufficient for solving \ksa{}. 
Basically, \MAJINF{k} guarantees that at most $k$ VSRCs exist in 
a run that are not affecting each other significantly. 
Despite being fairly strong,
the resulting message adversary \MAkdI{d} allows
to implement a $k$-uniform $k$-set agreement algorithm, which 
naturally respects partitions and is \emph{worst-case $k$-optimal},
in the sense that no algorithm can solve $k-1$-set agreement
under \MAkdI{d}. To the best of our
knowledge, it is the first gracefully degrading 
consensus algorithm proposed so far.
\end{enumerate}

As a final remark, we note that the ultimate goal of the latter 
part of our research 
are network assumptions for every $1\leq k < n$, 
which are both necessary and sufficient for solving $k$-set agreement.
Knowing or at least approaching this border is interesting for several 
reasons:
First, it is interesting from a theoretical point of view:
$k$-set agreement has been a major target for the study 
of solvability in asynchronous systems with failure detectors since decades.\footnote{Despite 
all efforts, however, the weakest failure detector for message-passing
$k$-set agreement is still unknown \cite{BR11:TCS}. Interestingly, 
\cite{RS13:PODC} revealed that there are relations between this classic model
and dynamic networks.} Second, striving for weak network assumptions is
always advantageous w.r.t.\ the assumption coverage in real systems, as
they are typically more likely to hold in a given dynamic network. Finally,
a set of network assumptions close to the necessary and sufficient ones
is needed for $k$-optimal $k$-set agreement algorithms:
Whereas our worst-case $k$-optimal algorithm only needs a single worst-case
run under \MAkdI{d} where it cannot solve $k-1$-set agreement,
a $k$-optimal algorithm must solve $k$-set agreement for the \emph{smallest} $k$ 
possible in \emph{every} run.

We believe that our work constitutes a significant step towards identifying the exact solvability border of $k$-set agreement:
Since necessary and sufficient network conditions in our model 
must lie somewhere 
in between (3) and (4), we managed to tightly ``enclose'' them.
Further tightening the gap and eventually closing it, is a topic
of future research.

\section{Related Work} 
\label{sec:relwork}

Dynamic networks have been studied intensively in research
(see the overview by Kuhn and Oshman \cite{KO11:SIGACT} and the references therein). Besides work on peer-to-peer networks like \cite{KSW10},
where the dynamicity of nodes (churn) is the primary concern, different
approaches for modeling dynamic connectivity have been proposed, both
in the networking context and in the context of classic distributed
computing. Casteigts et al.~\cite{CFQS12:TVG} introduced a comprehensive
classification of time-varying graph models.

\noindent
\textbf{Models.} There is a rich body of literature on dynamic graph models going
     back to~\cite{HG97:Dyn}, which also mentions for the first time
     modeling a dynamic graph as a sequence of static graphs. A more 
recent paper using this approach is \cite{KLO10:STOC}, 
where distributed computations are organized 
in lock-step synchronous rounds. Communication is described by a sequence
of per-round communication graphs, which must adhere to certain
network assumptions (like $T$-interval connectivity,
which says that there is a common subgraph in any interval of $T$ rounds).
Afek and Gafni \cite{AG13} introduced message adversaries for specifying
network assumptions in this context, and used them for 
relating problems solvable in wait-free read-write shared memory systems
to those solvable in message-passing systems. 
Raynal and Stainer \cite{RS13:PODC} also used message adversaries for
exploring the relationship between round-based models and failure 
detectors. 

Besides time-varying graphs, several alternative approaches that 
consider missing
messages as failures have also been proposed in the past: Moving omission
failures~\cite{SW89}, round-by-round fault detectors~\cite{Gaf98},
the heard-of model~\cite{CBS09} and the perception-based failure
model~\cite{BSW11:hyb}.

\noindent
\textbf{Agreement problems.}
Agreement problems in dynamic networks with undirected 
communication graphs have been studied in \cite{KOM11,APRU12,CRW11}; 
agreement in directed graphs has been considered in 
\cite{SWK09,BRS12:sirocco,AG13,RS13:PODC,CG13}.

In particular, the work by Kuhn et al.~\cite{KOM11} focuses on the $\Delta$-coordinated 
consensus problem, which extends consensus by requiring all processes to decide within
$\Delta$ rounds of the first decision. Since they consider only 
undirected graphs that are 
connected in every round, without node failures, solving consensus 
is always possible.
In terms of the classes of~\cite{CFQS12:TVG}, the model of
\cite{kuhn+lo:dynamic} is in one of the strongest classes (Class~10)
in which every process is always reachable by every other process. 
On the other hand, \cite{SWK09,CG13} do consider directed graphs, but restrict the dynamicity
by not allowing stabilizing behavior. Consequently, they also belong to 
quite strong classes of network assumptions in \cite{CFQS12:TVG}. In sharp
contrast, the message adversary tolerated by our algorithms does not
guarantee bidirectional (multi-hop) communication between all
processes, hence falls between the weakest and second weakest class of models
defined in~\cite{CFQS12:TVG}.

The leader election problem in dynamic networks
has been studied in \cite{CRW11,CRW10:dialm-pomc}, where the adversary controls the mobility of nodes in a
wireless ad-hoc network. This induces dynamic changes of the (undirected)
network graph in every round and requires any leader election algorithm to
take $\Omega(D n)$ rounds in the worst case, where $D$ is a bound on
information propagation.

Regarding $k$-set agreement in dynamic networks, we are not aware of
any previous work except \cite{SS14:DISC}, where bidirectional links are assumed,
and our previous paper~\cite{BRS11:IPDPS}, where
we assumed the existence of an underlying \emph{static} skeleton graph
(a non-empty common intersection of the communication graphs of all rounds)
with at most $k$ \emph{static} root components. Note that this essentially
implies a directed dynamic network with a static core. By contrast, in
this paper, we allow the directed communication graphs to be fully dynamic.
In \cite{BRS13:TPDS},
we provided $k$-set agreement algorithms for partially synchronous
systems with weak synchrony requirements.

\noindent
\textbf{Degrading consensus problems.} 
We are also not aware of related work exploring gracefully degrading consensus
or $k$-uniform
$k$-set agreement. However, there have been several attempts to
weaken the semantics of consensus, in order to cope with partitionable
systems and excessive faults. Vaidya and Pradhan introduced the notion
of \emph{degradable} agreement \cite{VP93}, where processes are allowed
to also decide on a (fixed) default value in case of excessive faults.
The \emph{almost everywhere agreement} problem introduced by \cite{DPPU88}
allows a small linear fraction of processes to remain undecided.
Aguilera et.~al.~\cite{ACT99} considered quiescent consensus in
partitionable systems, which requires processes outside the majority
partition not to terminate. None of these approaches is comparable
to gracefully degrading $k$-set agreement, however: 
On the one hand, we allow more
different decisions, on the other hand, all correct processes are required to
decide and every decision must be the initial value of some process.

Ingram et.~al.~\cite{ISWW09:IPDPS} presented an asynchronous leader election
algorithm for dynamic systems, where every component is guaranteed to elect a
leader of its own. Whereas this behavior clearly matches our definition of
graceful degradation, contrary to decisions, leader assignments are revocable
and the algorithm of \cite{ISWW09:IPDPS} is guaranteed to
successfully elect a leader only once the topology eventually stabilizes.

\section{Model} \label{sec:model}

We consider a synchronous distributed system made up of a fixed set of
distributed processes $\Pi=\{p_1,\dots,p_n\}$ with $|\Pi|=n\geq 2$, which
have fixed unique ids and communicate via unreliable message passing. 
For convenience, we assume that the unique id of $p_i \in \Pi$ is $i$, and 
use both $p_i$ and $i$ for denoting this process; 
``generic'' processes will also be denoted by $p$, $q$ etc.

Similar to the $\mathcal{LOCAL}\/$ model \cite{Pel00}, we assume that 
processes organize their computation as an infinite
sequence of communication-closed~\cite{EF82} 
lock-step rounds. For every $p\in\Pi$ and each round
$r>0$, let $S_p^r \in \S_p$ be the state of $p$ at the beginning of
round $r$, taken from the set  $\S_p$ of all states $p$ can possibly enter; 
$S_p^1\in \S_p^1\subset\S_p$ is taken from the set of $p$'s initial states
$\S_p^1$. The round $r$ 
computation of process $p$ is determined by two functions 
that make up $p$'s algorithm: The message sending function $M_p:\S_p \to 
\M$ determines the message $m_p^r$, taken from a suitable message alphabet $\M$, 
sent to all other processes in the system by $p$ in round $r$, based on 
$p$'s state $S_p^{r}$ at the beginning of round $r$. For simplicity,
we assume that some (possibly $\mbox{NULL} \in \M$) message is sent to all 
in a round where there is no proper algorithm message to be broadcast.
A receiver may omit to receive a message sent to it in a round, and
senders do not know (without receiving explicit feedback later on)
who successfully received their message.
The transition function $T_p:\S_p\times 2^{(\Pi\times\M)} \to \S_p$ takes 
$p$'s state 
$S_p^{r}$ at the beginning of round $r$ and a set $\mu_p^r$ of pairs of process
ids and messages, which contains the round $r$ 
messages received by $p$ from other processes in the system, 
and computes the successor state $S_p^{r+1}$. 
We assume that, for each process $q$, there is at most one
$(q,m_q^r)\in\mu_p^r$ such that $m_q^r$ is the message $q$ sent in round $r$.
Note that neither $M_p$ nor $T_p$ need to involve $n$, i.e., the algorithms
executed by the processes may be uniform with respect to the network size $n$.

\tikzstyle{p}=[circle,draw=gray,fill=lightgray!30,thick,inner sep=0pt,minimum size=4mm]
\tikzstyle{p6}=[p,double]
\tikzstyle{link}=[->,black,thick,auto]
\tikzstyle{fail}=[->,black,thick,densely dotted,auto]
\tikzstyle{i}=[draw=none,opacity=0]
\newcommand{\showproc}[5]{%
  \node[#1]  (p1)               {$p_1$};
  \node[#3]  (p3) [right=of p1] {$p_3$};
  \node[#5]  (p5) [right=of p3] {$p_5$};
  \node[#4]  (p4) [below=of p3] {$p_4$};
  \node[#2]  (p2) [above=of p3] {$p_2$};

  \begin{pgfonlayer}{background}
  \draw[draw=none,fill=lightgray!20] (current bounding box.south west) rectangle (current bounding box.north east);
  \end{pgfonlayer}
}  
\setlength{\textfloatsep}{\baselineskip}
\begin{figure*}[t]
  \tiny
\centering{
\subfloat[$\Ga^1$]{%
\begin{tikzpicture}
  \showproc{p}{p}{p}{p}{p}

  \draw[link] (p1) to [bend left] (p2);
  \draw[link] (p2) to [bend left] (p1);
  \draw[link] (p4) to (p1);
  \draw[link] (p4) to (p5);
  \draw[link] (p2) to (p3);
  \draw[link] (p5) to (p2);
\end{tikzpicture}
\label{fig:g1}
}\ 
\subfloat[$\Ga^2$]{%
\begin{tikzpicture}
  \showproc{p}{p}{p}{p}{p}

  \draw[link] (p1) to (p2);
  \draw[link] (p2) to (p3);
  \draw[link] (p4) to (p1);
  \draw[link] (p4) to (p5);
\end{tikzpicture}
\label{fig:g2}
}\ 
\subfloat[$\Ga^3$]{%
\begin{tikzpicture}
  \showproc{p}{p}{p}{p}{p}

  \draw[link] (p2) to (p1);
  \draw[link] (p3) to (p1);
  \draw[link] (p5) to (p3);
  \draw[link] (p5) to (p2);
  \draw[link] (p3) to (p4);
\end{tikzpicture}
\label{fig:g3}
}
\caption{
\label{fig:graphs}%
}
} %
\end{figure*}

The evolving nature of the network topology is modeled as an infinite
      sequence of simple directed graphs $\G^1,\G^2,\dots$, which is determined 
      by an omniscient \emph{message adversary} \cite{RS13:PODC,AG13} that has access to the 
processes' states.

\begin{definition}[Communication graphs]\label{def:commgraphs}
For each round $r$, the \emph{round $r$ communication graph} $\Gr=\li{V,\Er}$
is a simple directed graph with node set $V=\Pi$ and edge set 
$\Er \subseteq \{(p\ra q) \colon p, q\neq p \in V\}$, where $(p\ra q) \in \Er$
\emph{iff} $q$ successfully receives $p$'s round $r$ message (in round $r$).
The set $\N_q^r$ denotes $q$'s \emph{in-neighbors} in $\Gr$ (excluding $q$).
\end{definition}
Note that we will sloppily write $(p\ra q) \in \Gr$ to denote $(p\ra q) \in E^r$,
as well as $p \in \Gr$ to denote $p\in V=\Pi$.

\cref{fig:graphs} shows a sequence of communication graphs for a 
network of $5$ processes, for rounds $1$ to $3$.
For deterministic algorithms, a run is completely determined by the initial
states of the processes and the sequence of communication graphs.
We emphasize that $p$ does not have any \emph{a priori} knowledge of its
neighbors, i.e., $p$ does not know who receives its round $r$
message, and does not know who it will receive from in round $r$
before its round $r$ computation.

Since every $\G^r$ can range arbitrarily from $n$ isolated nodes 
to a fully connected graph, there is no hope to solve any
non-trivial agreement problem without restricting the power
of the adversary to drop messages\footnote{Even though the adversary 
can only affect communication in our model,
it is also possible to model classic send and/or receive omission
process failures \cite{PT86} (and
thereby also crash failures): A process that is send/receive
omission faulty in round $r$ has no outgoing/incoming edges
to/from some other processes in $\Gr$.} 
to some extent.
Inspired by \cite{RS13:PODC}, we encapsulate a particular restriction, 
e.g., that every communication graph must be strongly connected, 
by means of a particular \emph{message adversary}. Note that \cref{def:ma}
generalizes the notation introduced in \cite{AG13}, which
just specified the \emph{set} of communciation graphs the adversary
may choose from in every round, to sets of \emph{sequences} of communication graphs.

\begin{definition}[Message adversary]\label{def:ma}
A message adversary \emph{Adv} (for our system $\Pi$ of $n$ processors) is a set
of sequences of communication graphs $(\G^r)_{r>0}$. A particular
sequence of communication graphs $(\A^r)_{r>0}$ is \emph{feasible} 
for \emph{Adv}, if $(\A^r)_{r>0} \in \mbox{\emph{Adv}}$.
\end{definition}
Informally, we say that some message adversary \emph{Adv} guarantees some
property, called a \emph{network assumption}, if every $(\G^r)_{r>0} \in \mbox{\emph{Adv}}$
satisfies this property.

For our system $\Pi$ of $n$ processes, this introduces a natural partial 
order of message adversaries, where $A$ is weaker than $B$ 
(denoted $A \leq B$) iff $A \subseteq B$, i.e., if it can generate at most the
communication graph sequences of $B$. 
As a consequence, an algorithm that works correctly under message 
adversary $B$ will also work under $A$.

\subsection{Consensus and $k$-set agreement}
\label{sec:problems}

To formally introduce the consensus and $k$-set agreement problem
studied in this paper, we assume some finite set $\V$ and consider the set
     of possible initial states $\S_p^1$ (of process $p$) to be
     partitioned into $|\V|$ subsets $\S_p^1[v]$, with $v\in \V$.
When $p$ starts in a state in $\S_p^1[v]$, we say that $v$ is $p$'s
input value, denoted $x_p=v$. Moreover, we assume that, for each $v\in
\V$, there is a set $\D_p[v]\subset\S_p$ of decided states such that
$\D_p[v]\cap\D_p[w]=\emptyset$ if $v\ne w$ and $\D_p[w]$ is closed under
$p$'s transition function, i.e., $T_p$ maps every state in this subset
to this subset (for all possible sets $\mu_p$ of received messages). We
say that $p$ has \emph{decided} on the output value (also called decision value)
$v$, denoted $y_p=v$, when it is in some state in $\D_p[v]$. When
$p$ performs a transition from a state outside of the set of decided
states to the set of decided states, we say that $p$ \emph{decides}.

\begin{definition}[Consensus]\label{def:consensus}
Algorithm $\mathcal{A}$ solves \emph{consensus}, if the following properties 
hold in every run of $\mathcal{A}$:
\begin{itemize}%
\item[(Agreement)] If process $p$ decides on $y_p$ and $q$ decides on
  $y_q$, then $y_p=y_q$.
\item[(Validity)] If $y_i=v$, then $v$ is some $p_j$'s initial value $x_j$.
\item[(Termination)] Every process must eventually decide. 
\end{itemize}
\end{definition}

For the \emph{$k$-set agreement problem}~\cite{Cha93}, we assume that both
$|\V|>k$ and $n > k$ to rule out trivial solutions.

\begin{definition}[$k$-set agreement]\label{def:ksa}
Algorithm $\mathcal{A}$ solves \emph{$k$-set agreement}, if the following properties 
hold in every run of $\mathcal{A}$:
\begin{itemize}%
\item[($k$-Agreement)] At most $k$ different decision values are obtained
system-wide in any run.
\item[(Validity)] If $y_i=v$, then $v$ is some $p_j$'s initial value $x_j$.
\item[(Termination)] Every process must eventually decide. 
\end{itemize}
\end{definition}
Clearly, consensus is the special case of $1$-set agreement; set agreement is a
short-hand for $n-1$-set agreement.

A consensus or $k$-set agreement algorithm is called \emph{uniform}, if it
does not have any a priori knowledge of the network (and hence of $n$).
A $k$-set agreement algorithm is called \emph{$k$-uniform}, if it does not 
require a priori knowledge of $k$.

\subsection{Basic network properties: Vertex-stable root components}
\label{sec:MAbasics}

We will now define the cornerstones of the message adversaries used in our 
paper, which culminate in \cref{def:D-bounded-VSRC} and 
\cref{def:E-network-vertex-stable-roots}. Message adversaries such as
\MAd{d} (\cref{ass:window}) and \MAkd{d} (\cref{ass:inter}) will be defined 
implicitly, by defining the properties of the sequences of feasible 
communication graphs.
Informally, most of those will rest 
on the pivotal concept of \emph{root components}, 
which are strongly connected components in $\Gr$ without \emph{incoming}
edges from processes outside the component. Our message adversaries
will be required to eventually guarantee root components that are 
vertex-stable, i.e., to consist of the
same \emph{set} of nodes (with possibly varying interconnect) during a sufficiently 
large number of consecutive rounds. Vertex-stability will eventually guarantee that
all members can receive information from each other.

\begin{definition}[Root Component] \label{def:RC}
A \emph{root component} $R^r$, with non-empty set of vertices $R\subseteq\Pi$,
is a \emph{strongly connected component} (SCC) in $\Gr$ that has no 
incoming edges from other components, formally
$\forall p \in R^r, \forall q \in \Gr\colon (q~\ra~p) \in \Gr 
\Rightarrow q \in R^r$.
\end{definition}
By contracting SCCs, it is easy to see that every weakly
connected directed simple graph $\G$ has at least one root component,
see \cref{lem:root}.
Hence, if $\G$ has $k$ root components, it has at most $k$
weakly connected components (with disjoint root components, but
possibly overlapping in the remaining processes).

\begin{definition}[Vertex-Stable Root Component] \label{def:VSRC}
A sequence of consecutive rounds with communication graphs $\G^x$ for $x\in 
I=[a,b]$, $b\geq a$, contains an $I$-\emph{vertex-stable root component} $R^I$,
if, for $x \in I$, every $\G^x$ contains a root component $R^x$ with the same set of nodes
$R$ (but possibly varying interconnection topology).
\end{definition}
We will abbreviate $R^I$ as an $I$-VSRC or $|I|$-VSRC if only the length
of $I$ matters, and sometimes denote an $I$-VSRC $R^I$ by its vertex set $R$
if $I$ is clear from the context. Note carefully that we assume $|I|=b-a+1$ here, since
$I=[a,b]$ ranges from the \emph{beginning} of round $a$ to the \emph{end} of round $b$; hence, $I=[r,r]$ is not empty but rather represents 
round $r$.

The most important property of a VSRC $R^I$ is that information is guaranteed
to spread to all its vertices $R$ if the interval $I$ is large enough, as proved in
\cref{lem:bncd} below. To express this formally, we need a few basic definitions
and lemmas. 

Similarly to the classic ``happened-before'' relation \cite{Lam78}, we say
that a process \emph{$p$ causally influences $q$ in round $r$}, denoted by
$(p \overset{r}{\leadsto} q)$, iff either (i) 
$q$ has an incoming edge
$(p \rightarrow q)$ from $p$ in $\G^r$, or (ii) if $q = p$, i.e., we assume that
$p$ always influences itself in a round. Given a sequence of communication graphs
$\G^r,\G^{r+1},\dots$, we say that there is an \emph{causal influence chain} 
of length $k \ge 1$ starting from $p$ in round $r$ to $q$, denoted by
$(p \overset{r[k]}{\leadsto} q)$, if there exists a
sequence of not necessarily distinct processes $p=p_0, \dots, p_k=q$ such that
$p_i\overset{r+i}{\leadsto}p_{i+1}$ for $0\le i < k$.
If $k$ is irrelevant, we just write $(p \overset{r}{\leadsto} q)$ or just $(p \leadsto q)$
and say that $p$ (in round $r$) causally influences $q$. 
This allows us to define the
notion of a dynamic causal distance between processes as given in \cref{def:dd}.

\begin{definition}[Dynamic causal distance]\label{def:dd}
Given a sequence of communication graphs $\G^r,\G^{r+1},\dots$,
the \emph{dynamic causal distance} $cd^r(p, q)$  from process $p$ (in round $r$) to process $q$ 
is the length of the shortest causal influence chain starting in $p$ in round $r$ and ending 
in $q$, formally $cd^r(p, q) := min\{k \colon (p \overset{r[k]}{\leadsto}q)\}$. We
define $cd^r(p, p) = 1$ and $cd^r(p,q) = \infty$ if $p$ never
influences $q$ after round $r$.
\end{definition}

Note that, in contrast to the similar notion of dynamic distance defined in
     \cite{KOM11}, the dynamic causal distance in our \emph{directed} graphs
     is not necessarily symmetric: If the adversary chooses the graphs $\G^r$ such that not all
     processes are strongly connected, the causal distance between two
     processes can even be finite in one and infinite in the other direction.
In fact, even if $\G^r$ is strongly connected for round $r$ (but not
     for rounds $r'>r$), $cd^r(p,q)$ can be infinite. However,
the following \cref{lem:cd} shows that the causal distance in successive rounds
cannot arbitrarily decrease.

\begin{lemma}\label{lem:cd}
Given a sequence of communication graphs $\G^r,\G^{r+1},\dots$,
for every two processes $p,q\in\Pi$ it holds that $\dist_{r+1}(p,q) \geq \dist_r(p,q)-1$. As a
consequence, if $\dist_r(p,q)=\infty$, then also $\dist_{r+1}(p,q)=\infty$.
\end{lemma}
\begin{proof}
Since $(p\lt p)$ in every round $r$, the definition of dynamic
causal distance trivially implies $\dist_r(p,q) \leq 1 + \dist_{r+1}(p,q)$. \qed
\end{proof}

Analogous to the dynamic diameter defined for undirected communication graphs in \cite{KOM11},
we now define the \emph{dynamic causal diameter} $\cdiam^x(R^I)$ for round $x$ in a 
$I$-VSRC $R^I$ as the largest round $x$ dynamic causal distance 
$cd^x(p,q)$ between any pair of processes $p, q \in R$:

\begin{definition}[Dynamic causal diameter]\label{def:cdiam}
Given a sequence of communication graphs $\G^r,\G^{r+1},\dots$,
let $I=[a,b]$, $r\leq a \leq b$, be a nonempty interval of indices in this 
sequence.\footnote{Note that we will implicitly assume that this sentence holds true
in the sequel when we write something like ``there is an interval $I=[a,b]$
with a VSRC $R^I$''.}
Assume that the subsequence of communication graphs $\G^x$ for $x\in I$ 
contains an $I$-VSRC $R^I$ with node set $R$. 
Then, the \emph{dynamic causal diameter} of $R^I$ for round $x$ is defined as
$\cdiam^x(R^I) :=  \max_{p,q \in R}\{cd^x(p,q)\}$.
\end{definition}

Obviously, it may be the case that $\cdiam^x(R^I)=\infty$ in general.
However, if $|I|$ is sufficiently large, the following \cref{lem:2} reveals 
that $\cdiam^x(R^I)<\infty$.

\begin{lemma}[Bound on dynamic causal diameter]\label{lem:2}
Given some $I=[a,b]$ and a VSRC $R^I$ with $|R| \ge 2$, if
$b \ge a+|R|-2$, then $\forall x \in [a,b-|R|+2] \colon \cdiam^x(R^I) \le |R|-1$.
\end{lemma}

\begin{proof}
Fix some process $p\in R$ and some $x$ where $a\leq x \leq 
b-|R|+2$.  Let $\P_0=\set{p}$, and define for each $i>0$ the set
$\P_i=\P_{i-1} \cup \{q: \exists q'\in\P_{i-1}: q'\in\N_{q}^{x+i-1}
\cap R\}$. 
$\P_i$ is hence the set of processes $q\in R$ such that $(p\ltedge{x[i]}q)$ holds. 
Using induction, we will show that  $|\P_k| \geq \min\{|R|,k+1\}$
for $k\geq0$. Induction base $k=0$: $|\P_0| \geq \min\{|R|,1\}=1$ follows immediately from
$\P_0=\set{p}$.
Induction step $k \to k+1$, $k\geq 0$:  
Clearly the result holds if $|\P_k|=|R|$, thus we
consider round $x+k$ and $|\P_k|<|R|$: 
It
follows from strong connectivity of $\Ga^{x+k} \cap R$ that there is a set of
edges from processes in $\P_k$ to some non-empty set
$\L_k\subseteq R\setminus\P_k$. Hence, we have
$\P_{k+1}=\P_k\cup\L_k$, which implies $|\P_{k+1}| \geq |\P_k| + 1 \geq k+1 + 1 = k+2 = \min\{|R|,k+2\}$
by the induction hypothesis.

Thus, in order to guarantee $R=\P_{k}$ and thus $|R|=|\P_{k}|$, choosing
$k$ such that $|R|=1+k$ and $k\le b-x+1$ is sufficient. 
Since $b\ge x+|R|-2$, both conditions can be fulfilled by choosing $k=|R|-1$.
Moreover, due to the definition of $\P_k$, it follows that $\dist_{x}(p,q)\le|R|-1$
for all $q\in R$. Since this holds for any $p$ and any $x\leq s-|R|+2$, 
the statement of \cref{lem:2} follows. \qed
\end{proof}

\cref{lem:2} thus implies that information available at any node $p\in R$
at the beginning of round $x \in [a,b-|R|+2]$ has spread to all other 
nodes in $R$ by the end of round $b$, i.e., during $I$. On the other hand,
it may be the case for some particular VSRC $R^I$ with $|I|< |R|-1$ 
that the information available at the beginning of some round $x \in I$ 
has already spread to all other nodes in $R$ by the end of round $b$.
\cref{lem:infprop} reveals that this implies that the information available
at any round $x' \in [r,x]$ has also been spread to all nodes in $R$
by the end of round $b$.

\begin{lemma}[Information propagation]\label{lem:infprop}
Suppose that $R^I$ for $I=[a,b]$ is an $I$-VSRC of size $|R|\ge 2$, such
that there is some $x\in [a,b]$ with $x+\cdiam^x(R^I)-1 \leq b$.  Then, 
for every $x'\in [a,x]$, it holds that $x'+\cdiam^{x'}(R^I)-1\leq 
  b$.
\end{lemma}
\begin{proof}
\cref{lem:cd} reveals that for all $p,q \in R^I$, we have
$x-1+\dist_{x-1}(p,q)-1 \leq x+\dist_x(p,q)-1\leq s$, which implies 
$x'+\dist_{x'}(p,q)-1\leq s$ for every $x'$ where $r\leq x' \leq x$ and 
proves our lemma.  
\qed
\end{proof}

Conversely, assume that some particular VSRC $R^I$ is such that information
available at the beginning of round $a$ reaches all members of $R$ by
the end of some round $a+D-1 < b$, i.e., $\cdiam^a(R^I) \leq D$
for some $D < |I|$. Can we infer something about $\cdiam^x(R^I)$
for later rounds $x>a$ in this case? In particular, will information available at the
beginning of round $b-D+1$ be spread to all nodes by the end of round $b$?
Unfortunately, in general, this is not the case, as the following simple example 
for $I=[1,2]$ and $|R|=3$ shows: If
$\Ga^1$ is the complete graph whereas $\Ga^2$ is a ring, $\cdiam^1(R^I)=D=1$,
but information propagation starting at round 2 does not reach all other nodes
by the end of of round $2$. 

This stimulated the following \cref{def:D-bounded-VSRC}, 
which parameterizes the worst-case information propagation in
a VSRC via a parameter $D$ that represents its dynamic causal diameter. Informally,
it guarantees that messages sent by any process in $R$, in any but the last $D-1$ rounds of 
$I$, reach all members of $R$ within $I$.

\begin{definition}[$D$-bounded $I$-VSRC]
  \label{def:D-bounded-VSRC}
An $I$-vertex-stable root component $R^I$ with $I=[a,b]$ is \emph{\goodD},
with \emph{dynamic causal diameter} $D>0$, if either $|I| < D$ \footnote{That is,
by convention, we also call a VSRC \goodD\ that is too short to be interesting.
Obviously, such a VSRC need not guarantee information propagation
within $D$ rounds. Note 
that it would actually be possible to write $|I|\leq D$ here, as our algorithms
do not even consider $D$-VSRCs as interesting; we chose the present definition for
consistency with \cref{def:E-network-vertex-stable-roots} for \goodDE\ VSRCs, however.} 
or else $\forall x \in [a,b-D+1] \colon \cdiam^x(R^I) \leq D$.
\end{definition}

\cref{lem:2} showed that every sufficiently long VSRC $R^I$ is $D$-bounded for 
$D\geq |R|-1$; all sufficiently long VSRCs are hence necessarily $(n-1)$-bounded.
On the other hand, choosing some $D<n-1$ can be used to force the message adversary to 
speed-up information propagation accordingly. For example, we show in \cref{sec:expander} that certain expander graph topologies ensure $D=O(\log n)$.

\medskip

To formalize information propagation from root components to the rest of the network,
one has to account for the fact that a process $q$ outside any root 
component may be reachable from \emph{multiple} root components in general. 
Intuitively speaking, this models dynamic networks that
do not ``cleanly'' partition. Given a sequence of communication 
graphs $\G^r,\G^{r+1},\dots$ containing a set $S^I=\{R_1^I,\dots,R_{\ell}^I\}$ 
of $\ell\geq 1$ $I$-VSRCs, all vertex-stable
in the same interval $I=[a,b]$, let the round $x$ \emph{dynamic network causal diameter} $\cheight^x$ be
the maximum, taken over all processes $q \in \Pi$,
of the minimal dynamic causal distance $cd^x(p,q)$ from \emph{some} process
$p\in \bigcup_{i=1}^\ell R_i^I$ in round $x$, formally
$\cheight^x(S^I):=\max_{q\in\Pi}\bigl\{\min_{p \in \cup_{i=1}^{\ell}R_i}\{cd^x(p,q)\}\bigr\}$.
\cref{def:E-network-vertex-stable-roots} will be used in the sequel
to guarantee that every process in the network receives a message
from some member
of at least one VSRC in $S^I=\{R_1^I,\dots,R_{\ell}^I\}$ within $\nwbound$ rounds
if $|I|\geq \nwbound$.

\begin{definition}[$\nwbound$-network-bounded $I$-VSRC]
\label{def:E-network-vertex-stable-roots}
A set $S^I=\{R_1^I,\dots,R_{\ell}^I\}$ of $\ell\geq 1$ $I$-VSRCs with
$I=[a,b]$ is $\nwbound$-network-bounded,
with \emph{dynamic network causal diameter} $\nwbound>0$, if either $|I|<\nwbound$
or else $\forall x \in [a,b-\nwbound+1]:\, \cheight^x(S^I) \leq H$.
\end{definition}

Note that \cref{def:E-network-vertex-stable-roots} guarantees $(p \overset{x[H]}{\leadsto} q)$
for \emph{at least one} but not for all $p\in R_i$. Moreover, $p$ (and hence
$R_i$) may be different for different starting rounds $x$ in $I$.

A comparison of \cref{def:E-network-vertex-stable-roots} and \cref{def:D-bounded-VSRC}
reveals that it always holds that $H\geq D$. Moreover, in the case $\ell=1$ (where $S^I$ contains 
a single root component $R^I$ only), \cref{def:E-network-vertex-stable-roots} is
exactly \cref{def:D-bounded-VSRC} with the dynamic causal diameter 
$\cdiam^x(R^I)$ replaced by the dynamic network causal diameter 
$\cheight^x(S^I)=:\cheight^x(R^I)$. Finally, analogous to \cref{lem:2},
the following \cref{lem:bncd} shows that the dynamic network causal diameter $\nwbound$
is bounded by $n-1$, provided $b-a \geq n-2$.

\begin{lemma}[Bound on dynamic network causal diameter]\label{lem:bncd}
Suppose there is some interval $I=[a,b]$ where there is a set
$S^I=\{R_1^I,\dots,R_{\ell}^I\}$
of exactly
$\ell\geq 1$ $I$-vertex-stable root components.
If $b\geq a+n-2$ and $n\geq 2$, then $S^I$ is $n-1$-network bounded.
\end{lemma}

\begin{proof}
Let $P_0=\bigcup_{i=1}^{\ell} R_i$ and fix
any $x$ where $a \le x \le b-n+2$. Define,
for each $i>0$, the set $P_i = P_{i-1} \cup \{q: \exists
q' \in P_{i-1}:q' \in \N_q^{x+i-1}\}$. $P_i$ is hence the set of processes $q$
such that $(p \overset{x[i]}{\leadsto} q)$ holds for at least one $p\in P_0$.
Using induction, we will show
that $|P_k| \ge \min\{n, k+1\}$ for $k \ge 0$. Induction start $k=0:|P_0|
\ge \min\{n,1\}=1$ follows immediately from $P_0 \supseteq \{p_1,\dots,p_\ell\}$ with
$\ell\geq 1$.
Induction step $k \rightarrow k+1,k \ge 0$: First assume that already
$|P_k|=n$; since $|P_{k+1}| \ge |P_k| = n \ge
\min\{n,k+2\}$, we are done. Otherwise, consider round $x+k$ and $|P_k| < n$:
Since every node $q\in\Pi$ is in a weakly connected component containing
at least one root in every round, hence also in $\G^{x+k}$, there is a set
of edges from processes in $P_k$ to some non-empty set $L_k \subseteq
\Pi\setminus P_k$.
Hence, we have $P_{k+1}=P_k \cup L_k$, which implies $|P_{k+1}| \ge |P_k|+1 \ge
k+1+1= k+2= \min\{n,k+2\}$ by the induction hypothesis.
Thus, in order to guarantee $\Pi=P_k$ and thus $n=|P_k|$, choosing $k$ such
that $n=1+k$ and $k \le b-x+1$ is sufficient. Since $b \ge x+n-2$,
both conditions can be fulfilled by choosing $k = n-1$. Moreover, due to the
definition of $P_k$, it follows that for all $q \in \Pi$ there is some $p\in
P_0$
with $cd^{x}(p,q) \le n-1$, implying $\cheight^{x} \le n-1$. Since this holds
for any
$x \le b-n+2$
following \cref{def:E-network-vertex-stable-roots}, this implies
\cref{lem:bncd}. \qed
\end{proof}

\subsection{An example for $\nwbound < n-1$: Expander topologies}
\label{sec:expander}

We conclude this section with an example of a network topology that guarantees a dynamic causal 
network diameter $\nwbound$ that is much smaller than $n-1$, which justifies why
we introduced this parameter (as well as $D$) explicitly in our model. 

An undirected graph $\G$ is an 
\emph{$\alpha$-vertex expander} if, for all sets $S \subset V(\G)$ of size $\le 
|V(\G)|/2$, it holds that $\frac{|\N(S)|}{|S|} \ge \alpha$, where $\N(S)$ is
the set of neighbors of $S$ in $\G$, i.e., those nodes in
$V(\G)\setminus S$ that have a neighbor in $S$. (Explicit expander constructions can be found in \cite{HLW2006:Expander}.)
As we need an expander property for \emph{directed} communication graphs,
we consider, for a vertex/process set $S$ and a round $r$,
both the set $\N^r_+(S)$ of nodes outside of $S$ that are reachable from $S$ and 
the set of nodes $\N^r_-(S)$ that can reach $S$ in $r$. \cref{ass:fast}
ensures an expansion property both for subsets $S$ chosen from root components 
(property (a)) and other processes (properties (b), (c)).

\begin{definition}[Directed Expander Topology] \label{ass:fast}
There is a fixed constant $\alpha$ and a fixed set $R$ such that the following 
conditions hold for all sets $S \subseteq V(\G^r)$:
\begin{compactenum}
\item[(a)] If $|S|\le |R|/2$ and $S\subseteq R$, then
  $\frac{|\N^r_+(S) \cap R|}{|S|} \ge \alpha$ and $\frac{|\N^r_-(S) 
  \cap  R|}{|S|} \ge \alpha$.
\item[(b)] If $|S|\le n/2$ and $R\subseteq S$, then
  $\frac{|\N^r_+(S)|}{|S|} \ge \alpha$.
\item[(c)] If $|S|\le n/2$ and $R \cap S = \emptyset$, then 
  $\frac{|\N^r_-(S)|}{|S|} \ge \alpha$.
\end{compactenum}
\end{definition}

The following \cref{lem:expgraphs} shows that (1) \cref{ass:fast} 
does not contradict the existence of
     a single root component and that (2) these expander
     topologies guarantee both a dynamic causal diameter $D=O(\log n)$ for
$I$-VSRCs with $|I|=O(\log n)$ and a dynamic causal network diameter $\nwbound=O(\log n)$.

\begin{lemma}\label{lem:expgraphs}
  There are sequences of graphs $(\Gr)_{r>0}$ with a single root component in
every $\Gr$ where \cref{ass:fast} holds and where, for any such run,
  there is an interval $I$ during which there exists a \goodD\ and 
$\nwbound$-network-bounded $I$-vertex stable root component with $D=O(\log n)$
and $\nwbound=O(\log n)$.
\end{lemma}

\begin{proof}
We will first argue that \emph{directed} graphs with a single root exist that
satisfy \cref{ass:fast}.
Consider the simple \emph{undirected} graph $\bar{\U}$ that is the 
union of an $\alpha$-vertex expander on $R^I$ with member set $R$, and
an $\alpha$-vertex expander on $V(\G^r)$. 
We turn $\bar{U}$ into a directed graph by replacing every edge $(p,q)\in E(\bar{\U})$ with oriented 
directed edges $p\ra q$ and $q\ra p$. This guarantees Properties
(a)-(c). In order to guarantee the existence of exactly one root
component, we drop all directed edges pointing to
$R^I$ from the remaining graph, i.e., we remove all edges  $p\ra q$
where $p\not\in R$ and $q\in R$, which leaves Properties (a)-(c)
intact and makes the $R$ from \cref{ass:fast} the single
root component of the graph.
We stress that the actual topologies chosen by the adversary might be 
quite different from this construction, which merely serves us to show 
the existence of such graphs.

We also recall that our message adversaries like the one given in
\cref{ass:window} will rely on vertex-stable root components $R^I$,
which only require that the set of its vertices $R$ remain unchanged, 
whereas the interconnect topology can change arbitrarily. Adding
\cref{ass:fast} does of course not change this fact. 

We will first show that the ``per round'' expander topology stipulated
by \cref{ass:fast} is strong enough to guarantee that every sufficiently long 
VSRC is \goodD\ with $D=O(\log n)$.

For $i\ge1$, let $\P_i\subseteq R$ be the set of processes $q$ in $R^I$ 
with $I=[a,b]$ and $|I|=O(\log n)$ such that $(p\ltedge{a[i]}q)$, and $\P_0=\{p\}$.
The result $D=O(\log n)$ follows immediately from \cref{lem:2}
if $|R|\in O(\log n)$, so assume that 
$|R|\in\Omega(\log n)$ and consider some process $p \in R$.
For round $a$, Property (a) yields $|\P_{1}| \ge |\P_{0}|(1+\alpha)$.
In fact, for all $i$ where $|\P_{i}|\le |R|/2$, we can apply Property (a) to 
get $|\P_{i + 1}| \ge |\P_{i}|(1+\alpha)$, hence $|\P_{i}| \ge
\min\{(1+\alpha)^i,|R|/2\}$. 
Let $\ell$ be the smallest value such that
$(1+\alpha)^{\ell}>|R|/2$, which guarantees that $|\P_{\ell}|>|R|/2$.
That is, $\ell = \left\lceil\frac{\log(|R|/2)}{\log(1+\alpha)}\right\rceil 
\in O(\log n)$.
Now consider any $q\in R$ and define $\Q_{i-1}\subset R$ as 
the set of nodes that causally influence the set $\Q_i$ in round $a+i$, for 
$\Q_{2\ell+1}=\{q\}$.
Again, by Property (a), we get $|\Q_{i-1}| \ge |\Q_i|(1+\alpha)$, so $|\Q_{2k-i}| \ge 
\max\{(1+\alpha)^{i},|R|/2\}$. From the definition of $\ell$ above, we
thus have $|\Q_{\ell}|>|R|/2$.
Since $\P_{\ell} \cap \Q_{\ell} \ne \emptyset$, it follows that every $p \in 
R$ influences every $q\in R$ within $2\ell\in O(\log n)$
rounds. While the above proof has been applied to the starting round $x=a$
only, it is evident that it carries over literally also for any 
$x<s-2\ell$, which shows that $R^I$ is indeed \goodD.

What remains to be shown is that $\nwbound$-network-boundedness with $\nwbound= O(\log
n)$ also holds. We use Properties (b) and (c) similarly as in the above proof:
For any round $x \in [r,s-2k']$, we know by (b) that any process $p\in R$ has 
influenced at least $n/2$ nodes by round $x+k'$ where $k' = \lceil 
\log_{1+\alpha}(n/2)\rceil \in O(\log n)$ by arguing as for the $\P_i$
sets above.
Now (c) allows us to reason along the same lines as for the sets 
$\Q_{i-1}$ above. That is, any $q$ in round $x+2k'$ will be influenced 
by at least $n/2$ nodes. Therefore, any $p$ will influence every $q\in\Pi$ by 
round 
$x+2k'$, which completes the proof. \qed
\end{proof}

This confirms that sequences of communication graphs with 
$D < n-1$ and $\nwbound < n-1$ indeed exists and are
compatible with message adversaries such as \MAd{d} stated in
\cref{ass:window} below.

\section{Consensus Impossibilities and Lower Bounds} \label{sec:consimposs}

In this section, we will prove that some a priori
knowledge of the dynamic network causal diameter and
the existence of a stable interval of a certain minimal size
are inevitable for soving consensus in our model. Moreover, we will introduce
the message adversary \MAd{d}, which will be shown in \cref{sec:consensus} to be 
weak enough for solving consensus if $d=2D+2\nwbound+2 \leq 4\nwbound+2$, albeit it is
too strong for solving other standard problems in dynamic networks like
reliable broadcasting.

Since consensus is trivially impossible for an unrestricted message
adversary, which may just inhibit any communication in the system, 
we start from a message adversary that guarantees weakly connected 
communication 
graphs $\Gr$ in every round $r$. However, it is not 
difficult to see that this not sufficient for solving consensus,
even when all $\Gr=\G$ are the same, i.e., in a static topology: 
Consider the case where $\G$ contains two root components $R_1$ and
$R_2$; such a graph obviously exists, cp.\ \cref{lem:root} below.
If all processes in $R_1$ start with initial value $0$ and 
all processes in $R_2$ start with initial value $1$, they must decide
on their own initial value and hence violate agreement. After all,
no process in, say, $R_1$ ever has an incoming link from any process not
in $R_1$.

We hence restrict our attention to message adversaries that guarantee
a \emph{single} root component in $\G^r$ for any 
round $r$. \cref{fig:graphs} showed a sequence of graphs where
this is the case. Some simple properties of such graphs
are asserted by \cref{lem:root}.

\begin{lemma}\label{lem:root}
Any $\Gr$ contains at least one and at most $n$ root components (isolated processes), 
which are all disjoint.
If $\Gr$ contains a single root component $R^r$, then 
$\Gr$ is weakly connected, and there is a directed 
(out-going) path from every $p\in R^r$ to every $q \in \Gr$.
\end{lemma}

\begin{proof}
We first show that every weakly connected directed simple graph $\G$ has at least
one root component. To see this, contract every SCC to a 
single vertex and remove all resulting self-loops. The resulting graph $\G'$ is a directed acyclic graph (DAG)
(and of course still weakly connected), and hence $\G'$ has at least one vertex
$R$ (corresponding to some SCC in $\G$) that has no incoming edges. By
construction, any such vertex $R$ corresponds to a root component in the
original graph $\G$. Since $\Gr$ has at least $1$ and at most $n$ weakly
connected components, the first statement of our lemma follows.

To prove the second statement, we use the observation that there is a directed
path from $u$ to $v$ in $\G$ if and only if there is a directed path from the vertex
$C_u$ (containing $u$) to the vertex $C_v$ (containing $v$) 
in the contracted graph $\G'$.
If there is only one root component in $\G$, the above observations imply
that there is exactly one vertex $R$ in the contracted graph $\G'$
that has no incoming edges.
Since $\G'$ is connected, $R$ has a directed path to every other vertex in $\G'$,
which implies that every process $p \in R$ has a directed path to every vertex
$q$, as required. \qed
\end{proof}

It follows from \cite{BRS11:IPDPS} that assuming a single root component
makes consensus solvable if the root component is static. In this paper, we
allow the root component to change throughout the run, i.e., the (single)
root component $R^r$ of $\G^r$ might consist of a different set of processes in
every round round $r$. However, it will turn out that a sufficiently
long interval of vertex-stability is indispensable for solving 
consensus in this setting. In the sequel, we will consider the message adversary
\MAd{d} stated in \cref{ass:window}, which implicitly enforces the
dynamic network causal diameter $\nwbound$ according to 
\cref{def:E-network-vertex-stable-roots} and
is parameterized by some stability window duration $d>0$. 

\begin{definition}[Consensus message adversary \MAd{d}] \label{ass:window}
The message adversary \MAd{d} is the set of all 
sequences of communication graphs $(\Gr)_{r>0}$, where
\begin{enumerate}
\item[(i)] for every round $r$, $\Gr$ contains exactly one root component $R^r$,
\item[(ii)] all vertex-stable root components 
occurring in any $(\Gr)_{r>0}$ are \goodDE,
\item[(iii)] for each $(\Gr)_{r>0}$, there exists some $r_{ST}>0$ and an interval 
of rounds $J=[r_{ST},r_{ST}+d-1]$ with a \goodDE\ $J$-vertex-stable root component.
\end{enumerate}
\end{definition}
Note that item~(ii) has been added to the above definition solely for the
sake of our consensus algorithm in \cref{sec:consensus}. All the impossibility
results and lower bounds in this section hold also when (ii) is dropped
or replaced by something (like \goodD\ VSRCs, as in \cref{ass:inter}) that 
does not affect item~(iii).

First, we relate the message adversary in \cref{ass:window} to 
the classification 
of~\cite{CFQS12:TVG}:  \cref{lem:classification}
reveals that it is stronger than the weakest class that requests
one node that eventually reaches all others, but weaker than the
second class that requests one node that is reached by all. By
contrast, models like \cite{KOM11,kuhn+lo:dynamic} that assume
bidirectionally connected graphs $\Gr$ in every round belong 
to the strongest classes (Class 10) in~\cite{CFQS12:TVG}.

\begin{lemma}[Properties of $\MAd{d}$]\label{lem:classification}
In every sequence $(\Gr)_{r>0}$ of communication graphs feasible for 
$\MAd{d}$, 
\begin{enumerate}
\item[(i)] there is at least one process $p$ such that $\dist_1(p,q)$ is
finite for all $q\in\Pi$, and this causal distance is in fact at most $n(n-2)+1$. 
\item[(ii)] Conversely, for $n>2$, the adversary can choose some sequence $(\Gr)_{r>0}$
where no process $p$ is causally influenced by all other processes 
$q$, i.e., ${\not\exists p}\ \forall q\colon \dist_1(q,p) < \infty$.
\end{enumerate}
\end{lemma}
\begin{proof}
\cref{ass:window} guarantees that there is (at most) one 
root component $R^r$ in every $\Gr$, $r>0$.
  Since we have infinitely many graphs in $(\Gr)_{r>0}$ 
but only finitely many
  processes, there is at least one process $p$ in
  $R^r$ for infinitely many $r$. Let $r_1,r_2,\dots$ be this sequence
  of rounds. Moreover, let $\P_0=\set{p}$, and define for each $i>0$ the set
  $\P_i = \P_{i-1} \cup \{q: \exists q'\in \P_{i-1}: q' \in \N_q^{r_i}\}$.

Using induction, we will show that $|\P_k| \geq \min\{n,k+1\}$ for $k\geq 0$. 
Consequently, by the end of round $r_{n-1}$ at latest, $p$ will have causally influenced all
processes in $\Pi$.
Induction base $k=0$: $|\P_0| \geq \min\{n,1\}=1$ follows immediately from
$\P_0=\set{p}$.
Induction step $k \to k+1$, $k\geq 0$:  First assume that already
$|\P_k|=n \geq \min\{n,k+1\}$; since $|\P_{k+1}|\geq
|\P_k|=n \geq \min\{n,k+1\}$, we are done.
  Otherwise, consider round $r_{k+1}$ and $|\P_k|<n$: Since $p$ is in $R^{r_{k+1}}$, there is a path from
  $p$ to any process $q$, in particular, to any process $q$ in
  $\Pi\setminus\P_k \neq \emptyset$. Let $(v\rightarrow w)$ be an edge on such a
  path, such that $v\in\P_k$ and $w\in\Pi\setminus\P_k$. Clearly, the existence of
  this edge implies that $v\in\N_{w}^{r_{k+1}}$ and thus $w\in\P_{k+1}$.
  Since this implies $|\P_{k+1}| \geq |\P_k| + 1 \geq k+1 + 1 = k+2 = \min\{n,k+2\}$
by the induction hypothesis,
  we are done.

Finally, at most $n(n-2)+1$ rounds are needed until all processes $q$
have been influenced by $p$, i.e., $r_{n-1}\leq n(n-2)+1$:
A pigeonhole argument reveals that at
least one process $p$ must have been in the root component for $n-1$ times
after so many rounds. After all, if every $p$ appeared at most $n-2$ times, we
could fill up at most $n(n-2)$ rounds. By the above result, this is
enough to secure that some $p$ influenced every $q$.

  The converse statement (ii) follows directly from
  considering a static star, for example, i.e., a communication graph where there is one central
     process $c$, and for all $r$, $\Gr=\li{\Pi,\set{(c\ra q)|
     q\in\Pi\setminus\set{c}}}$. 
Clearly, $c$ cannot be causally influenced by any other process, and $q\not\lt q'$
for any $q,q'\neq q \in \Pi\setminus\set{c}$. On the other hand,
this topology satisfy \cref{ass:window}, which
includes the requirement of at most one root component per round. \qed
\end{proof}

Next, we examine the solvability of several broadcast problems \cite{kuhn+lo:dynamic}
under the message
adversary of \cref{ass:window}, summarized in \cref{thm:impossibleProblems}.
Although there is a strong bond between some of these problems and consensus in 
traditional settings, they are \emph{not} implementable under our 
assumptions---basically, because there is no guarantee of
(eventual) bidirectional communication. 

\begin{theorem} \label{thm:impossibleProblems}
Under the message adversary \MAd{d} given in 
\cref{ass:window}, for any $d$,
neither \emph{reliable broadcast}, \emph{atomic broadcast}, nor 
\emph{causal-order broadcast} can be implemented.
Moreover, there is no algorithm that solves
\emph{counting}, \emph{$k$-verification}, \emph{$k$-token
dissemination}, \emph{all-to-all token dissemination}, and 
\emph{$k$-committee election}.
\end{theorem}

\begin{proof}
We first consider reliable broadcast,
which requires that when a correct process broadcasts
$m$, every correct process eventually delivers $m$. Suppose
that the adversary chooses the communication graphs $\forall r:\ 
\Gr=\li{\set{p,q,s},\set{(p\ra q),(q\ra s)}}$, which matches \cref{ass:window}.
Clearly, $q$ is a correct process in our model. Since $p$ never receives a message from $q$,
$p$ can trivially never deliver a message that $q$ broadcasts.

For the token dissemination problems stated in \cite{kuhn+lo:dynamic}, 
consider the same communication graphs and assume that there is a token that 
only $s$ has. 
Since no other process ever receives a message from $s$, token dissemination is
     impossible.

For counting, $k$-verification, and $k$-committee election,
we return to the static star round graph $\Gr=\li{\Pi,\set{(c\ra q)|
     q\in\Pi\setminus\set{c}}}$ with central node $c$ considered in the
proof of \cref{lem:classification}.
As the local history of any process is obviously independent of $n$ 
here, it is impossible to solve any of these problems.
 \qed
\end{proof}

\subsection{Necessity of a priori knowledge of the dynamic network causal diameter}
\label{sec:impossDiameter}

We will now show that every correct solution for 
consensus, as well as for the related leader-election problem, requires
some a priori knowledge of the dynamic network causal diameter of
the communication graphs generated by the adversary. Recall that a uniform
algorithm does not have any priori knowledge of the network,
i.e., does not even know upper bounds for the dynamic network causal diameter $\nwbound$
(and hence for $n$).

\begin{theorem}[Impossibility of uniform consensus]  \label{thm:impossDiameter}
There is no uniform algorithm that can solve consensus under 
the message adversary \MAd{d} given in \cref{ass:window}, for
any $d$.
\end{theorem}

\begin{proof}
Assume for the sake of a contradiction that there is such a uniform 
algorithm $\A$, w.l.o.g.\ for a set of input values $\V$ that contains 0 and 1.
Consider a run $\alpha_v$ of $\A$ on a communication graph $\G$
that forms a (very large) static directed line rooted at process $p$
and ending in process $q$.
Process $p$ has initial value $v\in [0,1]$, while all other processes have
initial value $0$.
Clearly, the uniform algorithm $\A$ must allow
$p$ to decide on $v$ by the end of round $\kappa$, where $\kappa$ is a constant
(independent of $\nwbound$ and $n$; we assume that $n$ is large enough to guarantee $n-1 >\kappa$).
Next, consider a run $\beta_{v}$ of $\A$ that has the same initial
states as $\alpha_v$, and communication graphs $(\B^r)_{r>0}$
that, during rounds $[1,\kappa]$, are also the same as in $\alpha_v$
(defining what happens after round $\kappa$ will be defered).
In any case, since $\alpha_v$ and $\beta_{v}$ are indistinguishable for $p$
until its decision round $\kappa$, it must also decide $v$ in $\beta_{v}$
at the end of round $\kappa$.

However, since $n > \kappa+1$, $q$ has not been causally influenced by $p$
by the end of round $\kappa$. Hence, it has the same state $S_p^{\kappa+1}$
both in $\beta_{v}$ and in $\beta_{1-v}$. As a consequence, it cannot
have decided by round $\kappa$: If $q$ decided $v$, it would violate 
agreement with $p$ in $\beta_{1-v}$. Now assume that runs $\beta_v$, $\beta_{1-v}$
are actually such that the stable window occurs later than round $\kappa$, 
i.e., $r_{ST}=\kappa+1$, and that the adversary just reverses
the direction of the line then: For all $\B^{\ell}$, $\ell\geq \kappa+1$, $q$ is 
the root and $p$ is the last process of the resulting topology.
Observe that the resulting $\beta_{v}$ still satisfies
\cref{ass:window}, since $q$ itself forms the only root component.  
Now, $q$ must eventually decide on some value $v'$ in some later round $\kappa'$, 
but since $q$ has been in the 
same state at the end of round $\kappa$ in both $\beta_{v}$
and $\beta_{1-v}$, it is also in the same state in round $\kappa'$
in both runs. Hence, its decision contradicts the decision of $p$ in 
$\beta_{1-v'}$. \qed
\end{proof}

We now use a more involved indistinguishability argument to show
that a slightly weaker problem than consensus, namely, leader election is also
impossible to solve uniformly under the message adversary 
$\MAd{d}$. The classic leader election problem 
(cf.\ \cite{Lyn96})
assumes that, eventually, exactly one process irrevocably elects itself as leader (by 
entering a special \textsc{elected} state) and every other process elects 
itself as non-leader (by entering the \textsc{non-elected} state).  
Non-leaders are not required to know the process id of the leader.

Whereas it is easy to achieve leader election in our model when consensus is
solveable, by just reaching consensus on the process ids in the system, the opposite
is not true: Since the leader elected by some algorithm need not be in the root component
that exists when consensus terminates, one cannot use the leader to 
disseminate a common value to all processes in order to solve
consensus atop of leader election.

\begin{theorem}[Impossibility of uniform leader election]
  \label{thm:impossDiameterLE}
There is no uniform algorithm that can solve leader election under 
the message adversary \MAd{d} given in \cref{ass:window}, for
any $d$.
\end{theorem}

\begin{proof}
We assume that there is a uniform algorithm $\A$ that solves the problem. 
Consider the execution $\alpha_w(m)$ of $\A$ in a static unidirectional 
chain of $m$ processes, headed by process $p$ with id $w$: Since $p$ has only a single out-going edge and does not
know $n$, it cannot know whether it has neighbors at all. Since it might
even be alone in the single-vertex graph consisting of $p$ only, it
must elect  itself as leader in any $\alpha_w(m)$, $m\geq 1$, 
after some $T_w$ rounds ($T_w$ may depend on $w$, however, as we do 
not restrict $\A$ to be time-bounded). 

Let $w$ and $z$ be two arbitrary different process ids, and let $T_w$ resp.\
$T_z$ be the termination times in the executions $\alpha_w(m)$
resp.\ $\alpha_z(m')$, for any $m$, $m'$; let $T=\max\{T_w,T_z\}$.

We now build a system consisting of $n=2T+3$ processes. To do so we 
assume a chain $\G_p$ of $T+1$ processes headed by $p$ (with id $w$)
and ending in process $t$, a second chain $\G_q$ of $T+1$ processes headed
by $q$ (with id $z$) and ending in process $s$, and the process $r$.

Now consider an execution $\beta$, which proceeds as follows: 
For the first $T$ rounds, the communication graph is the
     unidirectional ring created by connecting the above chains with
     edges $(s\ra p)$, $(t\ra r)$ and $(r\ra q)$; 
its root component clearly is the entire ring. 
Starting from round $T+1$ on, process $r$ forms the single vertex root
     component, which feeds, through edges $(r\ra q)$ and $(r\ra t)$ the
     two chains $\G_q$ and $\bar{\G}_p$, with $\bar{\G}_p$ being $\G_p$ with all
     edges reversed. 
Note that, from round $T+1$ on, there is no edge connecting processes
     in $\G_p$ with those in $\G_q$ or vice versa.

Let $\ell$ be the process that is elected leader in $\beta$.
We distinguish 2 cases: 

\begin{enumerate}
\item If $\ell \in \G_q \cup \{r\}$, then
consider the execution $\beta_p$ that is exactly like $\beta$, except
that there is no edge $(s\ra p)$ during the first $T$ rounds:
$p$ with id $w$ is the single root component here. Clearly, for
$p$, the execution $\beta_p$ is indistinguishable from $\alpha_w(2T+3)$
during the first $T_w\leq T$ rounds,
so it must elect itself leader. However, since no process in $\G_q \cup \{r\}$ (including $t=\ell$)
is causally influenced by $p$ during the first $T$ rounds, all
processes in $\G_q\cup\{r\}$ have the same
state after round $T$ (and all later rounds) in $\beta_p$ as in $\beta$.  
Consequently, $\ell$
also elects itself leader in $\beta_p$ as it does in $\beta$, which is a 
contradiction.
\item On the other hand, if $\ell \in \G_p$, we consider the execution
$\beta_q$, which is exactly like $\beta$, except
that there is no edge $(r\ra q)$ during the first $T$ rounds:
$q$ with id $z$ is the single root component here. Clearly, for
$q$, the execution $\beta_q$ is indistinguishable from $\alpha_z(T+1)$ (made
up of the chain $\G_q$) during the first $T_z\leq T$ rounds,
so it must elect itself leader. However, since no process $t$ in $\G_p \cup \{r\}$ (including $t=\ell$)
is causally influenced by $q$ during the first $T$ rounds, $t$ has the 
same
state after round $T$ (and all later rounds) in $\beta_q$ as in $\beta$.  
Consequently, $\ell$
also elects itself leader $\beta_q$ as it does in $\beta$, which is again 
a contradiction.
\end{enumerate}
This completes the proof of \cref{thm:impossDiameterLE}. \qed
\end{proof}

\subsection{Impossibility of consensus with too short stability intervals}

The goal of this section is to show that root components $R^I$ must be
vertex-stable sufficiently long for solving consensus in our model. In essence,
what is needed for this purpose is that every member of the set $R$ of processes
in $R^I$ is able to reach the entire network. Recalling
\cref{def:E-network-vertex-stable-roots}, this requires $|I|$ to be
at least $\nwbound$ and hence $d \geq \nwbound$ in 
\cref{ass:window}.

In order to show that \MAd{\nwbound} is indeed necessary in our setting, we will now consider a
stronger message adversary \MAdp{\nwbound-1} given in \cref{ass:too-short} below: 
It is stronger than \MAd{\nwbound} as its stability interval is shorter, but
still slightly weaker than $\MAd{\nwbound-1}$, in that it also guarantees one process 
to be reached from the processes in $R$ within $\nwbound$ rounds, despite the too
short stability interval $I$.
Note carefully that, since there is only one such process, it would be reached if 
$|I|$ was actually $\nwbound$. 
This property is formally captured by \emph{\almostgoodDE} VSRCs introduced in
\cref{def:almostE-network-vertex-stable-roots}, which is slightly weaker than
\cref{def:E-network-vertex-stable-roots} in that $I$-VSRC's with $|I|=\nwbound-1$
are no longer arbitrary.

\begin{definition}[Almost $\nwbound-1$-bounded $I$-VSRC]
\label{def:almostE-network-vertex-stable-roots}
An $I$-vertex-stable root component $R^I$ with $I=[a,b]$ is \emph{\almostgoodDE},
with \emph{dynamic network causal diameter} $\nwbound>0$, if either $|I|<\nwbound-1$ or else 
$\forall x \in [a,b-\nwbound+2]$
there exists a unique $q\in\Pi$ with $\forall p\in R: \dist_x(p,q)\le \nwbound$, while for
all $q'\in\Pi\setminus\set{q}$ we have $\forall p\in R: \dist_x(p,q')\le \nwbound-1$.
\end{definition}

\begin{definition}\label{ass:too-short}
The message adversary \MAdp{\nwbound-1} is the set of all 
sequences of communication graphs $(\Gr)_{r>0}$, where
\begin{enumerate}
\item[(i)] for every round $r$, $\Gr$ contains exactly one root component $R^r$,
\item[(ii)] all vertex-stable root components $R^I$ 
occurring in any $(\Gr)_{r>0}$ are \goodDE,
\item[(iii)] for each $(\Gr)_{r>0}$, there exists some $r_{ST}>0$ and an interval 
of rounds $J=[r_{ST},r_{ST}+\nwbound-2]$ with an \almostgoodDE\ $J$-vertex-stable root component.
\end{enumerate}
\end{definition}

Note carefully that \cref{ass:too-short} allows the message adversary 
to choose \emph{any} communication graph sequence that is consistent with the
conditions stated therein. In particular, \MAdp{\nwbound-1} can choose a sequence of
communication graphs that ensures a dynamic causal distance $\nwbound$ 
between \emph{any} specific $p\in R^I$ and $q$ in a VSRC with $|I|=\nwbound-1$. 
Moreover, we have the following \cref{lem:advrelation} that relates our
message adversaries:

\begin{lemma}\label{lem:advrelation}
It holds that $\MAd{\nwbound-1} \geq \MAdp{\nwbound-1} \geq \MAd{\nwbound}$, so
that every sequence of communication graphs generated by the message adversary
\MAd{\nwbound} is also feasible for \MAdp{\nwbound-1}.
\end{lemma}
\begin{proof}
A comparison of \cref{ass:too-short} and \cref{ass:window} reveals that they
differ only in item~(iii). Since \almostgoodDE\ is slightly weaker than \goodDE,
as the adversary needs to guarantee a network causal distance $\dist_x(p,q')$
of at most $H-1$ from every $p\in R$ to every $q'\neq q$ in the former, 
$\MAd{\nwbound-1} \geq \MAdp{\nwbound-1}$ follows: After all, $\MAd{\nwbound-1}$
assumes a \goodDE\ VSRC. On the other hand, 
\cref{ass:too-short} does not forbid the message adversary to generate a
sequence of communication graphs that adheres to \cref{ass:window} with $d=\nwbound$, which
also confirms $\MAdp{\nwbound-1} \geq \MAd{\nwbound}$ and completes our proof. 
\qed
\end{proof}

We will now prove that the message adversary \MAdp{\nwbound-1}, and hence
by \cref{lem:advrelation} also \MAd{\nwbound-1}, is too strong
for solving consensus: Processes can \emph{withold} information
from each other, which causes consensus to be impossible \cite{SWK09}.
In order to simplify our proof, we assume that the adversary has to fix the
start of $J=[r_{ST},r_{ST}+\nwbound-2]$ and the set of root members $R$ in the 
eventually generated root component $R^J$ before the beginning of the 
execution (but given the initial values). 
Note that this does not strengthen the adversary, and hence does not 
weaken our impossibility result: For 
deterministic algorithms, the whole execution depends only on the
initial values and the sequence of the $\G^r$'s, so the
adversary could simulate the execution and determine
every $\G^{r+1}$ based on this. 

\begin{lemma}\label{lem:neighbours-are-bivalent}
Consider two runs of a consensus algorithm $\A$ under message adversary
\MAdp{\nwbound-1}, for some a priori fixed $J=[r_{ST},r_{ST}+\nwbound-2]$ and 
set of processes $R$ in $R^J$, which
start from two univalent configurations $C'$ and $C''$ that differ 
only in the state of one process $p$ at the beginning 
of round $r$. Then, $C'$ and $C''$ cannot differ in valency.
\end{lemma}

\begin{proof}
The proof proceeds by assuming the contrary, i.e., that $C'$ and $C''$
have different valency. We will then apply the same sequence of round
graphs to extend the execution prefixes that led to $C'$ and $C''$ to
get two different runs $e'$ and $e''$. It suffices to show that
there is at least one process $q$ that cannot distinguish $e'$ from
$e''$: This implies that $q$ will eventually decide on the same value in both
executions, which contradicts the assumed different valency of $C'$
and $C''$. 

Our choice of the round graphs depends on the following exhaustive 
cases: 
\begin{enumerate}
\item[(i)] For $p\not\in R$, we let the adversary choose any root component
$R^s$ consisting of the processes in $R$, for all $s\geq r$. Obviously, every process (i.e.,
we can choose any) $q\in R$ 
has the same state throughout $e'$ and $e''$.
\item[(ii)] For $p\in R$ and $r\in J$, we choose any root component $R^s$ consisting of the
processes in $R$ for $r\leq s \leq r_{ST}+\nwbound-2$, and $R^s=\{q\}$ for $s>r_{ST}+\nwbound-2$,
where $q$ is the process that does not hear from any process in
$R$ (and hence from $p$) within $J$ according to 
\cref{def:almostE-network-vertex-stable-roots}. Hence, $q$ has the same
state in $e'$ and $e''$, both during $J$ and afterwards, where
it is the single root.
\item[(iii)]  For $p\in R$ and $r\not\in J$, we choose graphs
$\G^s$ where $R^s=\{q\}$ and $p$ has only in-edges for 
$r\leq s < r_{ST}$; $q$ (satisfying $q\not\in R$ and hence $q\neq p$) 
is 
again the ``distant'' process allowed by 
\cref{def:almostE-network-vertex-stable-roots}. From $s=r_{ST}$ on, we
choose the same graphs $\G^s$ as in case (ii). It is again
obvious that $q$ has the same state throughout  $e'$ and $e''$,
since $p$ cannot communicate to any process before $J$ and
does not reach $q$ within $J$.
\end{enumerate}
In any case, for process $q$, the sequence of states in the extensions
starting from $C'$ and $C''$ is hence the same. Therefore, the two
runs are indistinguishable for $q$, which cannot hence decide 
differently. This provides the required contradiction to the
different valencies of $C'$ and $C''$. \qed
\end{proof}

The next \cref{lem:graph-seq} establishes connectedness
of the successor graph of a configuration \cite{SWK09}. 

\begin{lemma}\label{lem:graph-seq}
For any two round 
$r$ graphs $\G'$ and $\G''$, we can find a finite
sequence of graphs $\G',\G_1,\dots\G_i\dots\G''$, each with a single
root component, where any two consecutive graphs differ only by at most one edge. We say that the configurations $C'$ resp.\ $C''$ reached by applying
  $\G'$ resp.\ $\G''$ to the same configuration $C$ are \emph{connected}
in this case.
Moreover, our construction guarantees that if the root components of
$\G'$ and $\G''$ consist of the same set of processes $R'=R''=R$,
the same is true for all $\G_i$.
\end{lemma}

\begin{proof} 
First, we consider two cases with respect to the members $R'$
and $R''$ of the respective root components: (a) $R'\cap R''=\emptyset$, (b)
$R'\cap R''\ne\emptyset$. Moreover, for the second part of the proof,
we also consider a special case of (b): (b') $R'=R''$.

For case (b) (and thus also for (b')), we consider $\G_1=\G'$.
For case (a), we construct $\G_1$ from $\G'$ as follows: Let $p'\in R'$
and $p''\in R''$, then $\G_1$ has the same edges as $\G'$ plus
$a=(p''\edge{}p')$, thus $R_1\supseteq R'\cup\set{p''}$ (recall that
$p''$ must be reachable from $R'$ already in $\G'$).
So, now we have that in both cases $\G'$ and $\G_1$ differ in at most
one edge. 
Moreover, there is a nonempty intersection between $R_1$ and $R''$.

In the first phase of our construction (which continues as long as $E''\setminus
E_i\ne\emptyset$), we construct $\G_{i+1}$ from $\G_{i}$, $i\ge1$, by
choosing one edge $e=(v\edge{}w)$ from $E''\setminus E_i$ and let
$\G_{i+1}$ have the same edges as $\G_i$ plus $e$.  Clearly, $\G_i$
and $\G_{i+1}$ differ in at most one edge. Moreover, when adding an
edge, we cannot add an additional root component, so as long as we add edges we
will have that $\G_{i+1}$ has a single root component $R_{i+1}\supset
R'$.  

When we reach a point in our construction where $E''\setminus
E_i=\emptyset$, the first phase ends. As $\G_i$ now contains all the
edges in $\G''$, i.e., $E_i\supset E''$, we have $R_i\supset
R''$. In the second phase of the construction, we remove edges. To this
end, we choose one edge $e=(v\edge{}w)$ from $E_i\setminus E''$, and
construct $\G_{i+1}$ from $\G_i$ by removing $e$. Again we have to
show that there is only one root component. Since we never remove
an edge in $E''$, $\G_i$ always contains a directed path
from some $x\in R''$ to both $v$ and $w$ that only uses edges in
$E''$. As $e\not\in E''$, this also holds for $\G_{i+1}$. Since there
is only one root component in $\G''$, this implies that there is only
one in $\G_{i+1}$.

Let $\G_j$ be the last graph constructed in the first phase, and
$\G_k$ the last graph constructed in the second phase. 
It is easy to see that $E_k=E_j\setminus(E_j\setminus E'')$, 
which implies that $E_k=E''$ and hence $\G_k=E''$. This completes
the proof of the first part of our lemma.

To see that the second part also holds, we consider case (b') 
in more detail and show by induction that $R_{i+1}=R_{i}=R$.
For
the base case, we recall that $\G_1=G'$ and thus $R_1=R'$. 
For the induction step, we consider first that the step involves
adding an edge $e=(v\edge{}w)$ (phase 1): Adding an edge 
can only modify the
root component when $v\not\in R_i$ and $w\in R_i$. Since such an edge $e$ is
not in $E''$ (as it has the same root component as $E'$), 
we cannot select it for addition, so the root component
does not change. If, on the other hand, the step from $\G_i$ to
$\G_{i+1}$ involves removing the edge $e=(v\edge{}w)$ (phase 2),
we only need to consider the case where $v\in R_i$. (If $v\not\in R_i$, then
also $w\not\in R_i$ so the root component cannot change by removing $e$.)
But since we never remove edges from $E''$, this implies that even
after removing $e$ there is still a path from $v$ to $w$, so the root
component cannot have changed. 
\qed
\end{proof}

The proof of the following impossibility result follows roughly 
along the lines of the proof of~\cite[Lemma~3]{SWK09}. It shows, by means of induction on the round number, that a
consensus algorithm $\A$ cannot reach a univalent configuration after any finite
number of rounds. 

\begin{theorem}[Impossibility of consensus under \MAd{\nwbound-1}]\label{thm:consensusimp}
There is no algorithm that solves consensus under the message adversary \MAdp{\nwbound-1}, 
and hence none under \MAd{\nwbound-1}.
\end{theorem}

\begin{proof}
We follow roughly along the lines of the proof of~\cite[Lemma
  3]{SWK09} and show per induction on the round number, that no
algorithm $\A$ can reach a univalent configuration by round $r$, 
for any $r>0$. Since no process can have decided in a bivalent
configuration, this violates the termination property of consensus.

For the base case, we consider binary consensus only and argue similar to~\cite{FLP85} but make use of
our stronger validity property: 
Let $C_x^0$ be the initial configuration, where the processes with
the $x$ smallest ids start with $1$ and all others with $0$. Clearly,
in $C_0^0$ all processes start with $0$ and in $C_n^0$ all start with
$1$, so the two configurations are $0$- and $1$-valent, respectively.
To see that for some $x$ $C_x^0$ must be bivalent, consider that this
is not the case, then there must be a $C_x^0$ that is $0$-valent while
$C_{x+1}^0$ is $1$-valent. But, these configurations differ only in
$p_{x+1}$, and so by \cref{lem:neighbours-are-bivalent} they
cannot be univalent with different valency.

For the induction step we assume that there is a bivalent
configuration $C$ at the beginning of round $r-1$, and show that there is
at least one such configuration at the beginning of round $r$.
We proceed by contradiction and assume all configurations at the
beginning of round $r$ are univalent. Since $C$ is bivalent and all
configurations at the beginning of $r$ are univalent, there
must be two configurations $C'$ and $C''$ at the beginning of round
$r$ which have different valency. 
Clearly, $C'$ and $C''$ are reached from $C$ by two different 
round $r-1$ graphs $\G'=\li{\Pi,E'}$ and
$\G''=\li{\Pi,E''}$. \cref{lem:graph-seq} shows that there is a
sequence of graphs such that $C'$ and $C''$ are connected. 
Each pair of subsequent graphs in this sequence
differs only in one link $(v\edge{}w)$, such that the resulting configurations
differ only in the state of $w$. Moreover, if the root component in
$\G'$ and $\G''$ is the same, all graphs in the sequence also have the
same root component. 
Since the valency of $C'$ and $C''$ was assumed to be
different, there must be two configurations $\overline{C}'$ and
$\overline{C}''$ in the corresponding sequence of configurations that have different 
valency and differ only in the state of one process, say $p$.
Applying \cref{lem:neighbours-are-bivalent} to  $\overline{C}'$ and
$\overline{C}''$ again produces a contradiction, and so not all
successors of $C$ can be univalent. 

We have hence established that \MAdp{\nwbound-1} is too strong for
consensus, which implies the same for \MAd{\nwbound-1} according to 
\cref{lem:advrelation}. \qed
\end{proof}

\section{A Consensus Algorithm for \MAd{2D + 2\nwbound+2}} 
\label{sec:consensus}

In this section, we show that it is possible to solve consensus
under the message adversary \MAd{2D+2\nwbound+2} given in \cref{ass:window}.

The underlying idea of our consensus algorithm is to use flooding to
     propagate the largest input value to everyone. 
However, as \cref{ass:window} does not guarantee
     bidirectional communication between every pair of processes according
to \cref{lem:classification}, flooding is not sufficient: 
The largest input value could be hidden at a single process $p$
     that never has outgoing edges. 
If such a leaf process $p$ would never accept smaller values, 
it is impossible to reach agreement (without potentially violating validity). 
Thus, we have to find a way to force $p$ to accept also a smaller
value. 

A well-known technique to do so is \emph{locking} a candidate value. Obviously, 
we do not want a leaf process to lock its value, but rather 
some process(es) that will be able to impose their locked value, i.e.,
can successfully flood the system. In addition, we may allow
processes that have successfully locked a value to decide only when
they are sure that every other process has accepted their value
as well. According to \cref{def:E-network-vertex-stable-roots}, both can be guaranteed
when these processes have been in a vertex stable root component long 
enough--- which is (amply) guaranteed by \MAd{2D+2\nwbound+2}. 

The first major ingredient of our consensus algorithm is 
a network approximation algorithm (described in \cref{sec:approxalgo}), 
which allows processes to detect their root membership in (past)
rounds. The core of our consensus algorithm (presented in 
\cref{sec:consensuscore}) then exploits this knowledge for reaching
agreement on locked values and imposes the resulting value on 
all processes in the network. As we will see, the main complication comes from the
fact that a process can detect whether it has been part of the
root component of round $r$ only with some latency.

\subsection{The Local Network Approximation Algorithm}
\label{sec:approxalgo}

According to our system model, no process $p$ has any initial 
knowledge of the network. In order to learn about VSRCs, 
for example, it hence needs to \emph{locally} acquire such
knowledge. Process $p$ achieves this by means of
\cref{alg:approx}, which maintains a \emph{network
estimate} $A_p$ in a local variable.\footnote{We denote the value of 
a variable $v$ of process $p$ in
round $r$ \emph{before} the round $r$ computation finishes as $v_p^r \in S_p^r \in \S_p$; we usually
suppress the superscript when it refers to the current round.}
$A_p$ is a graph that holds the local estimates of every communication
graph $\G^r$ that occurred so far, simply by labeling an edge $(p \to q)$ with the set of
round numbers of every $\G^r$ once $p$ received evidence that $(p \to q)$ was present in
round $r$.

Initially, $A_p$ consists of process $p$ only. In every round, every process $p$ 
broadcasts its current $A_p$ and fuses it with the network estimates received 
from its neighbors. In more detail, $p$ updates $A_p$ whenever $q \in \N_p^r$, 
by adding $(q\edge{\{r\}}p)$ if $q$ is $p$'s neighbour for the first time, or by
updating the label of the edge $(q\edge{U}p)$ to $(q\edge{U \cup
 \{r\}}p)$ (\cref{line:addEdge1} and~\cref{line:addEdge2}).
Moreover, $p$ also receives $A_q$ from $q$ and uses this information
to update its own knowledge: 
The loop in \cref{line:forloop} ensures that $p$ has an edge
$(v\edge{T\cup T'}w)$ for each $(v\edge{T'}w)$ in $A_q$, where 
$T$ is the set of rounds previously known to $p$.

Given $A_p$, we use $A_p|t$ with $0$\footnote{To simplify the presentation, we have refrained from purging outdated
information from the network approximation graph. Actually, our consensus algorithm 
only queries $\stableSCC$ for intervals that span at most the
last $2\nwbound+1$ rounds, i.e., any older information could safely be removed from 
the approximation graph, resulting in a message complexity that is polynomial in
$n$.}$< t \le r$ to denote the current estimate of
$\G^t$ contained in $A_p$. Formally, $A_p|t$ is the graph induced by
the set of edges
\[
    E_p|t=\set{e=(v\rightarrow w)\mid \exists T\supseteq\set{t}:
      (v\edge{T} w)\in A_p}.
\]
As the information about $q$'s neighbors in $\G^t$ might take many rounds
  to reach some process $p$ (if it ever arrives at $p$), $A_p|t$
     may never be fully up-to-date, and as only reported
     edges are added to the estimate (but not all reports need to reach
     $p$), $A_p|t$ will be an under-approximation of $\G^t$. For example, 
a process $p$ that does not have any incoming links from other processes,
throughout the entire run of the algorithm, cannot learn
anything about the remaining network, i.e., $A_p$ will permanently be
the singleton graph.

\cref{alg:approx} finally provides an externally callable function $\stableSCC(I)$, 
which will be used by the core consensus consensus algorithm to find out
whether the calling process $p$ was member in an $I$-VSRC $R^I$ and to query
the set of all members $R$. We will prove in \cref{lem:Cpr2root} below that the latter
is the case if $A_p|t$ is strongly connected and consists of the same non-empty
set $R$ of processes for all $t\in I$. Informally, this is due to the fact
that the members of an $I$-VSRC will not be able to acquire knowledge of 
the topology outside $R^I$ within $I$, as they do not have incoming links
from outside.

\begin{algorithm}[!h]
  \footnotesize
\caption{\em Local Network Approximation (Process~$p_i$)}
\label{alg:approx}
\setlinenosize{\footnotesize}
\setlinenofont{\tt}
\begin{algorithmic}[1]
\item[] Provides externally callable function $\stableSCC()$. %
\EMPTY
\item[] {\bf Variables and Initialization:}

\STATE $A_{p_i}:=\li{V_{p_i},E_{p_i}}$ initially $(\set{p_i},\emptyset)$
\COMMENT{weighted digraph without multi-edges and loops}
\item[] {\bf\boldmath Emit round $r$ messages:}
\STATE send $\msg{A_{p_i}}$ to all current neighbors
\item[] {\bf\boldmath Round $r$: computation:}
\FOR{$q \in \Timely_{p_i}^r$ and $q$ sent message $\msg{A_q}$ in $r$}
  \IF{$\exists$ edge $e=(q\edge{T}p_i) \in E_{p_i}$}
    \STATE replace $e$ with $(q\edge{T'}p_i)$ in $E_{p_i}$
    where $T' \la T \cup \set{r}$ \label{line:addEdge1}
  \ELSE
    \STATE add $e:=(q\edge{\set{r}}p_i)$ to $E_{p_i}$ \label{line:addEdge2}
  \ENDIF
  \STATE $V_{p_i} \la V_{p_i} \cup V_q$
\ENDFOR
\FOR{every pair of nodes $(v,w)\in V_{p_i}\times V_{p_i}$, $v\neq w$}
\label{line:forloop}
  \IF{$T'=\bigcup\set{ S \mid \exists q\in\Timely_{p_i}^r\colon (v\edge{S}w)\in
E_q} \neq \emptyset$}
    \STATE replace $(v\edge{T}w)$ in $E_{p_i}$ with $(v\edge{T\cup
    T'}w)$; add $(v\edge{T'}w)$ if no such edge exists
  \ENDIF
  \ENDFOR \label{line:forloopend}
\EMPTY
\FUNC{$\stableSCC(I)$}
  \STATE Let $A_{p_i}|t$ be induced graph of $\set{(v\edge{T}w)\in E_{p_i} \mid t \in T}$
  \STATE Let $C_{p_i}|t$ be $A_{p_i}|t$ if it is strongly connected, 
or the empty graph otherwise.\label{line:Cps}
 \IF{$\forall t_1,t_2 \in I: C_{p_i}:=V(C_{p_i}|t_1)=V(C_{p_i}|t_2)\neq\emptyset$}
   \STATE return $C_{p_i}$
 \ELSE
   \STATE return $\emptyset$
\ENDIF
\ENDFUNC
\end{algorithmic}%
\end{algorithm}

We start our analysis of \cref{alg:approx} with 
\cref{lem:AsubsetG}, which shows that $A_p|t$ underapproximates 
$\G^t$ in a way that consistently
includes neighborhoods. Its proof uses the trivial invariant
asserting
$A_p|t = \li{\{p\},\emptyset}$ at the end of every round $r<t$.

\begin{lemma} \label{lem:AsubsetG}
If $A_p|t$ contains $(v\edge{}w)$ at the end
of some round $r$, then (i) $(v\edge{}w)\in\G^t$, i.e., $A_p|t\subseteq\G^t$,
and (ii) $A_p|t$ also contains $(v'\edge{}w)$ for every $v'\in 
\N_w^t \subseteq \G^t$.
\end{lemma}

\begin{proof} 
We first consider the case where $r<t$, then at the end of round $r$
$A_p|t$ is empty, i.e., there are no edges in $A_p|t$. As the
precondition of the Lemma's statement is false, the statement is true.

For the case where $r\geq t$, we proceed by induction on $r$:

Induction base $r=t$: If $A_p|t$ contains $(v\edge{}w)$ at the end of 
round $r=t$,
it follows from $A_q|t = \li{\{q\},\emptyset}$ at the end of every 
round $r<t$, for every $q\in\Pi$, that $w=p$, since $p$ is the only
processor that can have added this edge to its graph approximation. 
Clearly, it did so only when $v\in \N_p^t$, i.e., 
$(v\edge{}w) \in \G^t$, and included also $(v'\edge{}w)$ for every 
$v'\in \N_p^t$ on that occasion. This confirms (i) and (ii).

Induction step $r\to r+1$, $r\geq t$: Assume, as our induction 
hypothesis, that
(i) and (ii) hold for any $A_q|t$ at the end of round $r$, 
in particular, for every
$q\in\N_p^{r+1}$. If indeed $(v\edge{}w)$ in $A_p|t$ at the end of 
round $r+1$, it must be contained in the union of round $r$
approximations
\[
  U=\left(A_p|t\right) \cup \left(\bigcup_{q\in \N_p^{r+1}} A_q|t\right)
\]
and hence in some $A_i|t$ ($i=q$ or $i=p$) at the end of round $r$.
Note that the edges (labeled $r+1$) added in round $r+1$ to $A_p$
are irrelevant for $A_p|t$ here, since $t < r+1$.

Consequently, by the induction hypothesis, $(v\edge{}w) \in \G^t$, thereby
confirming (i). As for (ii), the induction hypothesis also implies that 
$(v'\edge{}w)$ is also in this $A_i|t$. Hence, every such
edge must be in $U$ and hence in $A_p|t$ at the end of 
round $r+1$ as asserted. \qed
\end{proof}

The following \cref{lem:Cpr2root} shows that locally detecting $A_p|t$ to be strongly connected
(in \cref{line:Cps} of \cref{alg:approx}) 
implies that $p$ is in the root component of round~$t$.
This result rests
on the fact that $A_p|t$ underapproximates $\G^t$ (\cref{lem:AsubsetG}.(i)), but does so in a way that
never omits an in-edge at any process $q\in A_p|t$ (\cref{lem:AsubsetG}.(ii)).

\begin{lemma}\label{lem:Cpr2root}
If the graph $C_p|t$ (\cref{line:Cps}) with $t<r$ is non-empty in round $r$,
then $p$ is member of $R^t$, i.e., $p\in R$.
\end{lemma}

\begin{proof}
For a contradiction, assume that $C_p|t$ is non-empty (hence $A_p|t$ is
an SCC by \cref{line:Cps}), but $p\not\in R$. Since $p$ is always
included in any $A_p$ by construction and $A_p|t$ underapproximates
$\G^t$ by \cref{lem:AsubsetG}.(i), this implies that $A_p|t$ cannot be
the root component of $\G^t$. Rather, $A_p|t$ must contain some process $w$ 
that has an in-edge $(v\edge{}w)$ in $\G^t$ that is not present in $A_p|t$.
As $w$ and hence some edge $(q\edge{t}w)$ is contained in $A_p|t$, because
it is an SCC, \cref{lem:AsubsetG}.(ii) reveals that this is impossible. \qed
\end{proof}

From the definition of the function $\stableSCC(I)$ in
\cref{alg:approx} and \cref{lem:Cpr2root}, we get the following
\cref{cor:stable2root}.

\begin{corollary}\label{cor:stable2root}
  If the function $\stableSCC(I)$
  evaluates to $R\neq\emptyset$ at process $p$ in round $r$, then $\forall x\in
  I$ where
  $x<r$, it holds that $p$ is a member of $R^{x}$, i.e., $p\in R$.
\end{corollary}

The following \cref{lem:root2Cpr} proves that, 
in a sufficiently
     long $I=[a,b]$ with a $I$-vertex-stable root component $R^I$, 
every member $p$ of $R^I$ detects an SCC for round $a$ (i.e., $C_p|a 
\neq \emptyset$) with a latency of at most $D$ rounds
(i.e., at the end of round $a+D$).
Informally speaking, together with \cref{lem:Cpr2root}, it
asserts that if there is an $I$-vertex-stable root component $R^I$ for
a sufficiently long interval $I$, then a process $p$ observes $C_p|a\neq 
\emptyset$ from the end of round $a+D$ on iff $p\in R$. 

\begin{lemma} \label{lem:root2Cpr}
Consider an interval of rounds $I=[a,b]$, such that there is a \goodD\ 
$I$-vertex-stable root component $R^I$ and assume $|I|=b-a+1 >D$. 
Then, from the end of round $a+D$ onwards, we have $C_p|a=R^I$,
     for every process in $p\in R^I$.
\end{lemma}

\begin{proof}
Consider any $q\in R^I$. At the beginning of round $a+1$, $q$ has
     an edge $(q'\edge{T}q)$ in its approximation graph $A_q$ with $a\in T$ iff
     $q'\in\N_q^a$. 
Since processes always merge all graph information from other
     processes into their own graph approximation, it follows from the definition of
     a \goodD\ $I$-vertex-stable root component (\cref{def:D-bounded-VSRC}) 
in conjunction with the fact that
     $a+1\leq b-D+1$ that every $p\in R^I$ has these
     in-edges of $q$ in its graph approximation by the end of round $a+1+D-1$. Since $R^I$
is a vertex-stable root-component, it is strongly connected without in-edges from
processes outside $R^I$. Hence $C_p|a=R^I$ from the end of round $a+D$ on, as asserted.
\end{proof}

This immediately gives us the following \cref{cor:allrootdec}, which ensures
that in a sufficiently long $I$-VSRC $R^I$, with $I=[a,b]$ and member set $R$, 
every $p\in R$ detects its membership in the $J$-VSRC
$R^J$, $J=[a,b-D] \subseteq I$, with a latency of at most $D$ rounds.

\begin{corollary}\label{cor:allrootdec}
Consider an interval of rounds $I=[a,b]$, with $|I|=b-a+1 > D$,
such that there is a \goodD\ vertex-stable root component $R^I$.
Then, from the end of round $b$ on, a call to $\stableSCC([a,b-D])$
returns $R$ at every process in $R$.
\end{corollary}

Together, \cref{cor:stable2root,cor:allrootdec} reveal that $\stableSCC(.)$
\emph{precisely} characterizes the caller's actual membership in the
$[a,b-D]$-VSRC $R^I$ in the communication graphs from the end of round $b$ on.

\subsection{Core consensus algorithm for \MAd{2D+2\nwbound+2}}
\label{sec:consensuscore}

\begin{algorithm}
  \footnotesize
\caption{Solving Consensus; code for process $p_i$}
\label{alg:consensus}
\setlinenosize{\footnotesize}
\setlinenofont{\tt}
\begin{algorithmic}[1]
\setcounter{ALC@unique}{0}
\STATE {Simultaneously run \cref{alg:approx}.}
\EMPTY
\item[] {\bf Variables and Initialization:}

\STATE $x_{p_i} \in \mathbb{N}$, initially own input value
\STATE $locked_{p_i}, decided_{p_i} \in \set{\texttt{false},\texttt{true}}$ initially
  $\texttt{false}$ 
\STATE $lockRound_{p_i} \in \mathbb{Z}$ initially $0$  
\EMPTY

\item[] {\bf\boldmath Emit round $r$ messages:}
\IF{$decided_{p_i}$}
\STATE send $\msg{\textsc{decide},x_{p_i}}$ to all neighbors         \label{line:sendDec}
\ELSE
  \STATE send $\msg{lockRound_{p_i},x_{p_i}}$ to all neighbors
\ENDIF
\EMPTY

\item[] {\bf\boldmath Round $r$ computation:} %
\IF{\textbf{not} $decided_{p_i}$} \label{line:secondIf}
  \IF{received $\msg{\textsc{decide},x_{q}}$ from any neighbor $q$}  \label{line:firstIf}
    \STATE $x_{p_i} \la x_{q}$
    \STATE decide on $x_{p_i}$ and set $decided_{p_i} \la \texttt{true}$\label{line:decOther}
  \EMPTY

  \ELSE[{$p_i$ only received $\msg{lock_{q},x_{q}}$ messages (if any)}:]
  \STATE $(lockRound_{p_i},x_{p_i}) \la \max\set{ (lock_{q},x_{q}) \mid q \in
    \Timely_{p_i}^r \cup \{p_i\}}$ \COMMENT{lexical order in $\max$} \label{line:updateXp}
  \IF{$\stableSCC([r-D-1,r-D]) \neq \emptyset$\label{line:if-G}}
    \IF{$(\text{\textbf{not} $locked_{p_i}$})$}
      \STATE $locked_{p_i}\la \texttt{true}$
      \STATE $lockRound_{p_i} \la r$ \label{line:lock}
    \ELSE 
      \IF{$\stableSCC([lockRound_{p_i},lockRound_{p_i}+\nwbound]) \neq \emptyset$\label{line:if-lock-heard}}
        \STATE decide on $x_{p_i}$ and set $decided_{p_i} \la \texttt{true}$
      \label{line:decideOwn} \label{line:decOwn}
      \ENDIF
    \ENDIF
  \ELSE[{$\stableSCC([r-D-1,r-D])$ returned $\emptyset$}]
    \STATE $locked_{p_i}\la \texttt{false}$
  \ENDIF
  \ENDIF
\ENDIF
\end{algorithmic}%
\end{algorithm}

As explained in \cref{sec:consensus}, the core consensus algorithm 
stated in \cref{alg:consensus}
builds upon the network approximation algorithm  given as \cref{alg:approx}: 
Relying on \cref{cor:stable2root},
every process uses $\stableSCC$ provided by \cref{alg:approx} to 
detect whether it has been in the vertex-stable root component of some past round(s).
Since \cref{cor:allrootdec} reveals that
$\stableSCC$ has a latency of up to $D \leq \nwbound$ rounds for reliably
detecting that a process is in the vertex-stable root component of some
(interval of) rounds, our algorithm (conservatively)
looks back $D$ rounds in the past when locking a value.

In more detail, \cref{alg:consensus} proceeds as follows: 
Initially, no process has locked a value, that is, $locked_p=\false$ and
     $lockRound_p=0$. 
Processes try to detect whether they are privileged by evaluating the
     condition in \cref{line:if-G}. 
When this condition is true in some round $\ell$, they lock the
     current value (by setting $locked_p=\true$ and $lockRound$ to the
     current round), unless $locked_p$ is already $\true$. 
Note that our locking mechanism does not actually protect the value
     against being overwritten by a larger value being also locked in $\ell$;
it locks out only those values that have older locks $l<\ell$.

When the process $m$ that had the largest value in the root component
     of round $\ell$ detects that it has been in a vertex-stable root component
     in all rounds $\ell$ to $\ell+\nwbound$ (\cref{line:if-lock-heard}), it
     can decide on its current value. 
As all other processes in that root component must have had $m$'s value imposed
     on them, they can decide as well.
     After deciding, a process stops participating in the
     flooding of locked values, but rather (\cref{line:sendDec})
     floods the network with $\msg{\textsc{decide},x}$. 
Since the stability window guaranteed by
     \cref{ass:window} with $d=2D+2\nwbound+2$ is large enough to allow every process
     to receive this message, all processes will eventually decide.

Before we turn our attention to the correctness proof of \cref{alg:consensus}, we need to
define how the network approximation algorithm and the core consensus
algorithm are combined to form a joint algorithm in our computation model.
Informally, we assume that (i) the complete round $r$ computing step
of the network approximation algorithm is executed just
before the round $r$ computing step of the consensus algorithm,
and that (ii) the round~$r$ message of the former is piggybacked
on the round~$r$ message of the latter. Consequently, the round $r$ computing
step of the consensus core algorithm, which terminates round~$r$,
can already access the \emph{result}
of the round $r$ computation of the network approximation algorithm,
i.e., its state at the \emph{end} of round $r$. Consequently,
\cref{cor:stable2root,cor:allrootdec} reveal that a call to 
$\stableSCC(I)$ with $I=[a,b-D]$ by $p$ in the transition function of 
round $b$ (or later) returns $\neq \emptyset$ \emph{precisely} when a VSRC $R^I$ containing
$p$ existed.

Formally, let $\S_p^N$, $\M^N$,
$T_p^N$, $M_p^N$ be the set of states, message alphabet, transition function, 
and message sending function of the network approximation algorithm,
with $S_p^{N,r} \in \S_p^N$, $m_p^{N,r} \in \M^N$ and $\mu_p^{N,r}$ denoting its state 
at the beginning of round $r$, the message sent in round $r$, and
the set of messages received in round $r$. 
Analogously, let $\S_p^C$, $\M^C$, $T_p^C$, $M_p^C$, $S_p^{C,r}$, $m_p^{C,r}$ and
$\mu_p^{C,r}$ be the corresponding entities for the core consensus
algorithm; note that $\S_p^N \cap \S_p^C = \{A_p\}$, albeit the core
consensus algorithm only reads (but never writes) the graph approximation $A_p$ (when
calling $\stableSCC$).

For the joint algorithm, we define the joint state space as
$\S_p^J = \S_p^N \cup \S_p^C$ and the joint message alphabet as
$\M^J = \M^N \times \M^C$. We assume that there are projection
functions $\pi^N: \S_p^J \to \S_p^N$ resp.\  $\pi^C: \S_p^J \to \S_p^C$
which, given $S_p^{J,r}$, can be used to obtain the
corresponding $S_p^{N,r}=\pi^N(S_p^{J,r})$ resp.\ $S_p^{C,r}=\pi^C(S_p^{J,r})$.
The joint message sending function $M_p^J: \S_p^J \to \M^J$
just computes the pair of messages $(m_p^{N,r},m_p^{C,r})$ via
$m_p^{N,r}=M_p^N(S_p^{N,r})$ and $m_p^{C,r}=M_p^C(S_p^{C,r})$.
The joint transition function 
$T_p^J: \S_p^J \times 2^{\Pi \times \M^J} \to \S_p^J$ first
applies $T_p^N$ to $S_p^{N,r}$ to compute (i) $S_p^{N,r+1}$
and (ii) an intermediate state $S_p^{C,r+}$ that is identical
to $S_p^{C,r}$ except that $A_p^r$ is replaced by the newly
computed $A_p^{r+1}$. $T_p^C$ is then applied to 
$S_p^{C,r+}$ to compute the state $S_p^{C,r+1}$, which finally
results in $S_p^{J,r+1}= S_p^{N,r+1} \cup S_p^{C,r+1}$. All this happens atomically and instantaneously at the round switching 
time.

\medskip

Our correctness proof starts with the validity property of consensus 
according to \cref{def:consensus}.

\begin{lemma}[Validity] \label{lem:validity}
  Every decision value is the input value of some process.
\end{lemma}

\begin{proof}
Processes decide either in \cref{line:decOther} or in
     \cref{line:decOwn}. 
When a process decides via the former case, it has received a
     $\msg{\textsc{decide},x_q}$ message, which is sent by $q$ iff $q$
     has decided on $x_q$ in an earlier round. 
In order to prove validity, it is thus sufficient to show that
     processes can only decide on some process' input value when they
     decide in \cref{line:decOwn}, where they decide on their
     current estimate $x_p$. 
Let the round of this decision be $r$. 
The estimate $x_p$ is either $p$'s initial value, or was updated in some round
     $r'\le r$ in \cref{line:updateXp} from a value received by
     way of one of its neighbors' $\msg{lockRound,x}$ message. 
In order to send such a message, $q$ must have had $x_q=x$ at the
     beginning of round $r'$, which in turn means that $x_q$ was
     either $q$'s initial value, or $q$ has updated $x_q$ after
     receiving a message in some round $r_q<r$.  
By repeating this argument, we will eventually reach a process
that sent its initial value, since no process can have updated its
     decision estimate prior to the first round. \qed
\end{proof}

The following \cref{lem:assorted-properties} states a number of 
properties maintained by our algorithm when the first process $p$ has 
decided.
Essentially, they say that there has been a vertex-stable root component 
in the interval $I=[\ell-D-1,\ell+\nwbound]$ centered around the lock 
round $\ell$ (but not earlier),
and asserts that all processes in that root component chose the same lock round $\ell$.

\begin{lemma}\label{lem:assorted-properties}
Suppose that process $p$ decides in round $r$, no decisions
     occurred before $r$, and $\ell=lockRound_p^r$,
then 
     \begin{itemize}
     \item[(i)] $p$ is
     in the vertex-stable root component $R^I$ with
     $I=[\ell-D-1,\ell+\nwbound]$ and member set $R$, 
     \item[(ii)] $\ell+\nwbound\le r\le\ell+\nwbound+D$,  
     \item[(iii)] $R\ne R'$, where $R'$ is the members set of the VSRC $R^{\ell-\nwbound-2}$, and  
     \item[(iv)] all processes in $R$
     executed \cref{line:lock} in round $\ell$, and no process
in $\Pi\setminus{R}$ can have executed \cref{line:lock} in
a round $\geq \ell$.
     \end{itemize}
\end{lemma}

\begin{proof}
Item (i) follows since \cref{line:if-G} has been continuously
      $\true$ since round $\ell$ and from \cref{lem:Cpr2root}.
As for item (ii), $\ell+\nwbound\le r$ follows from the requirement of
      \cref{line:if-lock-heard}, while $r\le\ell+\nwbound+D$ follows from
      (i) and the fact that by \cref{lem:root2Cpr} the
      requirement of \cref{line:if-lock-heard} cannot be, for the
      first time, fullfilled strictly after round $\ell+\nwbound+D$. 
From \cref{lem:root2Cpr}, it also follows that if
      $R=R'$, then the condition in \cref{line:if-G}
      would return true already in round $\ell-1$, thus locking would occur already in round
      $\ell-1$. 
Since $p$ did not lock in round $\ell-1$, (iii) must hold. Finally, 
from (i), (iii), and \cref{lem:root2Cpr}, it follows that
      every other process in $R$ also has
      $\stableSCC([\ell-D-1,\ell-D])=\true$ in round $\ell$. 
Moreover, due to (iii), $\stableSCC([\ell-1-D-1,\ell-1-D])=\false$ 
in round $\ell-1$, which causes
      all the processes in $R$ (as well as those in $\Pi\setminus{R}$)
to set $lockRound$ to 0. 
Since $\stableSCC([\ell'-D-1,\ell'-D])$ cannot become true for any 
$\ell'\geq \ell$ at a process $q\in \Pi\setminus R$, as $C_q|r=\emptyset$
for any $r\in I$ by \cref{cor:stable2root}, (iv) also holds. \qed
\end{proof}

The following \cref{lem:highlander} asserts that if a process
decides, then it has successfully imposed its proposal value on all other 
processes.

\begin{lemma}[Agreement] \label{lem:highlander}
Suppose that process $p$ decides in \cref{line:decideOwn} in round $r$ and
that no other process has executed \cref{line:decideOwn} before $r$.
Then, for all $q$, it holds that $x_q^r = x_p^r$.
\end{lemma}

\begin{proof}
Using items (i) and (iv) in \cref{lem:assorted-properties}, we can
     conclude that $p$ was in the vertex-stable root component of
     rounds $\ell=lockRound_p^r$ to $\ell+\nwbound$ and that all processes  in
     it member set $R$ have locked in round $\ell$. 
Therefore, in the interval $[\ell,\ell+\nwbound]$, $\ell$ is the maximal
     value of $lockRound$. 
More specifically, all processes $q$ in $R$ have $lockRound_q=\ell$,
     whereas all processes $s$ in $\Pi\setminus R$ have
     $lockRound_s<\ell$ during these rounds by \cref{lem:assorted-properties}.(iv).
Let $m\in R$ have the largest proposal value
     $x_m^\ell=x_{max}$ among all processes in $R$. 
Since $m$ is in $R$, there is a causal chain of length at most $\nwbound$
     from $m$ to any $q\in\Pi$. Note carefully that guaranteeing 
this property requires item~(ii)
of \cref{ass:window}, as the first decision (in round $r$) need not occur in the 
eventually guaranteed $2D+2\nwbound+2$-VSRC but already in some earlier ``spurious'' 
VSRC.

Since no process executed \cref{line:decideOwn} before round $r$, no
     process will send $\textsc{decide}$ messages in $[\ell,\ell+\nwbound]$. 
Thus, all processes continue to execute the update rule of
     \cref{line:updateXp}, which implies that $x_{max}$ will
     propagate along the aforementioned causal path to $q$. \qed
\end{proof}

\begin{theorem}[Consensus under \MAd{2D+2\nwbound+2}] \label{thm:consensus}
  Let $r_{ST}$ be the beginning of the stability window guaranteed by the message adversary
\MAd{2D+2\nwbound+2} given in \cref{ass:window}. Then, \cref{alg:consensus} in conjunction with 
\cref{alg:approx} solves consensus by the end of round $r_{ST}+2D+2\nwbound+1$.
\end{theorem}
\begin{proof}
  Validity holds by \cref{lem:validity}. Considering \cref{lem:highlander}, 
  we immediately get agreement: Since the first process $p$ that decides must
  do so via \cref{line:decOwn}, there are no other proposal values left
  in the system.

  Observe that, so far, we have not used the liveness part of 
  \cref{ass:window}.  In
  fact, \cref{alg:consensus} is always safe in the sense that agreement
  and validity are not violated, even if there is no vertex-stable root
  component. 

  We now show the termination property. By
  \cref{cor:allrootdec}, we know that every process in $p\in R$
  evaluates the predicate $\stableSCC([r_{ST},r_{ST}+1])=\true$ in
  round $\ell=r_{ST}+D+1$, thus locking in that round.
  Furthermore, \cref{ass:window} and \cref{cor:allrootdec}
  imply that at the latest in round $d=\ell+D+\nwbound$ every process $p\in R$ will evaluate the
  condition of \cref{line:if-lock-heard} to $\true$ and thus
  decide using \cref{line:decOwn}.
  Thus, every such process $p$ will send out a message
  $m=\msg{\textsc{decide},x_p}$. 
By \cref{def:E-network-vertex-stable-roots} and \cref{ass:window}, we know that
every $q\in\Pi$ will receive a $\textsc{decide}$ message at the latest in round $d+\nwbound=\ell+D+2\nwbound=r_{ST}+2D+2\nwbound+1$ and decide by the end of this round. \qed
\end{proof}

\section{Impossibilities and Lower Bounds for $k$-Set Agreement }
\label{sec:impossibility-proofs}

In this section, we will turn our attention from consensus to
general \ksa{} and prove related impossibility results and lower bounds.
We will accomplish this by showing that certain ``natural'' message 
adversaries do not allow to solve \ksa{}. For example, as
excessive partitioning of the system into more than $k$ root 
components makes $k$-set agreement trivially impossible,
one natural assumption is to restrict the maximum number of 
root components per round in our system to $k$. 

\cref{ass:inter} below defines the generic message adversary \MAkd{d}, which allows
at most $k$ VSRCs per round and guarantees a common window of
vertex stability of duration at least $d$. Note that it implicitly
involves both the dynamic causal diameter $D$ and the dynamic 
network causal diameter $\nwbound \geq D$ according to \cref{def:D-bounded-VSRC} and
\cref{def:E-network-vertex-stable-roots} (that have be enforced by
the message adversary).

\begin{definition}[Message adversary \MAkd{d}] \label{ass:inter}
The message adversary \MAkd{d} is the set of all sequences of 
communication graphs $(\Gr)_{r>0}$, where
\begin{enumerate}
\item[(i)] for every round $r$, $\Gr$ contains at most $k$ root components,
\item[(ii)] all vertex-stable root components 
occurring in any $(\Gr)_{r>0}$ are \goodD,
\item[(iii)] for each $(\Gr)_{r>0}$, there exists some $r_{ST}>0$ and an interval 
of rounds $J=[r_{ST},r_{ST}+d-1]$ where $1\leq \ell \leq k$ \goodDE\ vertex-stable root components
$R_1^J,\dots,R_{\ell}^J$ exist simultaneously.
\end{enumerate}
\end{definition}
Like for \cref{ass:window}, item~(ii) has only been added for the sake
of the $k$-set agreement algorithm (\cref{alg:ksa}); the impossibility
results and lower bounds also hold when (ii) is dropped or replaced
by something that does not affect item~(iii). Observe
that \MAdv{1}{d} is the same as \MAd{d} except that item~(ii) requires
all VSRCs to be \goodD\ instead of \goodDE.
Note also that the message adversary \MAkd{1} guarantees at most 
$k$ VSRCs in every $\Gr$, $r>0$.

We will now prove that it is impossible to solve \ksa{} for $1\leq k < n-1$ under
the message adversary
\MAkd{\min\{n-k,\nwbound\}-1}, even under the slightly weaker 
version of this message adversary stated in \cref{thm:D-static-roots} below.
We will use the generic impossibility theorem provided in 
\cite[Thm.~1]{BRS11:OPODIS} for this purpose. In a nutshell, the latter exploits
the fact that $k$-set agreement is impossible if $k$ sufficiently
disconnected components may occur and consensus cannot be solved in some component. 

We first introduce the required definitions:
Two executions of an algorithm $\alpha, \beta$ are
\emph{indistinguishable} (until decision)
for a set of processes $\D$, denoted $\alpha \indist{\D} \beta$, if for any $p
\in \D$ it holds that $p$ executes the same state transitions in $\alpha$ and
in $\beta$ (until it decides).
Now consider a model of a distributed system $\M = \li{\Pi}$ that consists of
the
set of processes $\Pi$ and a \emph{restricted model} $\M' = \li{\D}$ that is
computationally
compatible to $\M$ (i.e., an algorithm designed for a process in $\M$ can be
executed on a process in $\M'$) and consists of the set of processes
$\D \subseteq \Pi$.
Let $\A$ be an algorithm that works in system $\M = \li{\Pi}$, where
$\M_\A$ denotes the set of runs of algorithm $\A$ on $\M$, and let $\D
\subseteq \Pi$ be a nonempty set of processes.
Given any restricted system $\M'
= \li{\D}$, the \emph{restricted algorithm} $\ARes$ for system $\M'$ is
constructed by dropping all messages sent to processes outside $\D$ in the
message sending function of $\A$.
We also need the following similarity relation between runs in
computationally compatible systems (cf.\ \cite[Definition 3]{BRS11:OPODIS}):
Let $\R$ and $\R'$ be sets of runs, and $\D$ be a non-empty set of processes.
We say that \emph{runs $\R'$ are compatible with runs $\R$ for processes in
$\D$}, denoted by $\R' \compat{\D} \R$, if $\forall \alpha \in \R'~\exists \beta
\in \R \colon \alpha \indist{\D} \beta$.

\begin{theorem}[$k$-Set Agreement Impossibility {\cite[Thm.~1]{BRS11:OPODIS}}]
\label{thm:impossibility}
Let $\Mod=\li{\Pi}$ be a system model and consider the runs $\M_A$ that are
generated
  by some fixed algorithm $A$ in $\Mod$, where every process starts
  with a distinct input value.
Fix some nonempty and pairwise disjoint sets of processes
$\D_1,\dots,\D_{k-1}$, and a
      set of distinct decision values $\set{v_1,\dots,v_{k-1}}$.
Moreover, let $\D=\bigcup_{1\le i< k}\D_i$ and $\PiD=\Pi\setminus \D$.
Consider the following two properties:
\begin{itemize}
\item[\boldmath\bf(dec-$\D$)]
For every set $\D_i$, value $v_i$ was proposed by
  some $p\in \D$, and there is some $q\in \D_i$ that decides $v_i$.
\item[\boldmath\bf(dec-$\PiD$)]
If $p_j \in \PiD$ then $p_j$ receives
  no messages from any process in $\D$ until every process in $\PiD$ has decided.
\end{itemize}
Let $\Rii\subseteq \M_A$ and $\Riii\subseteq \M_A$ be the sets
  of runs of $A$ where (dec-$\PiD$) respectively both, (dec-$\D$) and
(dec-$\PiD$), hold.\footnote{Note that
  $\Rii$ is by definition compatible with the runs of the restricted
algorithm $\APiD$.}
Suppose that the following conditions are satisfied:
\begin{description}
\item[\bf\textup\condNonempty]
$\Rii$ is nonempty.
\item[\bf\textup\condRiiRiii]
$\Rii\subruns[\PiD]\Riii$.
\end{description}
In addition, consider a restricted model $\Mod'=\li{\PiD}$ such that the
following properties hold:
\begin{description}
\item[\bf\textup\condNocons]
There is no algorithm that solves consensus in $\Mod'$.
\item[\bf\textup\condMADMA]
$\MAD'\subruns \M_A$.
\end{description}
Then, $A$ does not solve $k$-set agreement in $\Mod$.
\end{theorem}

The proof of \cref{thm:D-static-roots} below
utilizes \cref{thm:impossibility} in conjunction with the 
impossibility of consensus under \MAd{\nwbound-1} established
in \cref{thm:consensusimp}.

\begin{theorem}[Impossibility of \ksa{} under \MAkd{\min\{n-k,\nwbound\}-1}] \label{thm:D-static-roots}
There is no algorithm that solves \ksa{} with $n > k+1$ processes under 
the message adversary \MAkd{\min\{n-k,\nwbound\}-1} stated in \cref{ass:inter}, 
for any $1 \leq k < n-1$, even if there are $k-1$
root components $R_1,\dots,R_{k-1}$ that are vertex-stable all the time, i.e., in $[1,\infty]$ 
(and only root component $R_k$ is vertex-stable for at most $\min\{n-k,\nwbound\}-1$ rounds).
\end{theorem}

\begin{proof}
Suppose that there is a $k$-set algorithm $\A$ that works correctly under the
assumptions of our theorem.
For $k=1$, \cref{thm:D-static-roots} is implied by \cref{thm:consensusimp},
since \MAdv{1}{H-1} is the same as \MAd{H-1} if item~(ii) is dropped in
both definitions. 

To prove the theorem for $k>1$,
we will show that the conditions of the generic \cref{thm:impossibility} are
satisfied, thereby providing a contradiction to the assumption that $\A$ exists.
Let $\D_i = \{p_i\}$ for $0 < i \leq k-1$
and let
$\D = \bigcup_{i=1}^{k-1} \D_i$.
Consequently, $\PiD$ = $\left\lbrace
p_k, p_{k+1}, \ldots, p_n \right\rbrace$ and $|\PiD|\geq 2$.

\noindent
\textbf{(A)} The set of runs $\mathbf{\Rii}$ of $\A$ where no process in $\PiD$
receives any message from $\D$ before it dedices is nonempty:
We choose the communication graph in every round to be such that $\PiD$ has 
no incoming links from $\D$ until every process in $\PiD$ has decided.
Since any such sequence of communication graphs satisfies the assumptions 
of our theorem, $\mathbf{\Rii} \neq \emptyset$.

\noindent
\textbf{(B)} The set of runs $\mathbf{\Riii}$ of $\A$ where both (i)
some process in every $D_i$ decides $v_i$ and (ii) no process in $\PiD$
receives any message from $\D$ before it decides satisfies
$\mathbf{\Rii \compat{\PiD} \Riii}$:
Let $\mH$ be the set of runs where processes $p_i$ have unique input values
$x_i = i$, $0 < i < k$,
the communication graph in every round is
such that $p_1, \ldots, p_{k-1}$ are isolated,
and $p_k, \ldots, p_n$ are weakly connected (with a single root)
until every process has
decided.
By the assumptions of our theorem, $\mH$ is non-empty.
Since (i) the processes in $\PiD$ never receive a message from a process in $\D$ in
both $\Rii$ and $\mH$,
and (ii) the initial values of the processes in $\PiD$ are not restricted in $\mH$ in any
way, it is easy to find, for any run $\rho \in \Rii$, a run $\rho' \in \mH$
such that $\rho \indist{\PiD} \rho'$.
Because obviously $\mH \subseteq \Riii$,
we have established $\Rii \compat{\PiD} \Riii$.

\noindent
\textbf{(C)} Consensus is impossible in $\mathbf{\M' = \left\langle \PiD
\right\rangle }$: Let
$\PiD$ be the partition containing the $k$\textsuperscript{th} root
component $R_k$, which is
perpetually changing in every round, except for some interval of rounds
$I = \left[  r_{ST}, r_{ST}+\ell-1 \right]$, where $\ell = \min\{n-k,\nwbound\}-1$,
for some fixed $r_{ST}$.
During this interval, let the topology of $\PiD$ be such that there exists some
$p \in R_k$ and some $q \in \PiD$ with $\dist_{r_{ST}}(p,q) = \ell+1$.
Since $|\PiD| = n-k+1$, such a topology
(e.g.\ a chain with head $p$ and tail $q$) can be created by the message
adversary \MAd{\nwbound-1} underlying \cref{thm:consensusimp} exists.
Hence, consensus is impossible in $\PiD$.

\noindent
\textbf{(D)} $\mathbf{\M'_{\A_{|\PiD}} \compat{\PiD} \M_\A}$:
Fix any run $\rho' \in \M'_{\A_{|\PiD}}$ and consider a run $\rho \in \M_\A$,
where every process in $\PiD$ has the same sequence of state transitions in $\rho$ as in
$\rho'$. Such a run $\rho$ exists, since the processes in $\PiD$ can be 
disconnected from $\D$ in every round in $\M_\A$, so $\rho \indist{\PiD} \rho'$. \qed
\end{proof}

Since \cref{thm:D-static-roots} tells us that
no $k$-set agreement algorithm (for $1 \leq k < n-1$) can
\emph{terminate} with insufficient concurrent
stability of the at most $k$ root components in the system, it is tempting to
assume that $k$-set agreement becomes solvable
if a round exists after which all communication graphs
remain the same.
However, we will prove in
\cref{thm:rGST-impossibility} below that this is
not the case for any $1 < k \leq n-1$. 
We will again use the generic \cref{thm:impossibility},
this time in conjunction with the variant of the well-known impossibility of consensus 
with lossy links \cite{SW89,SWK09} provided in \cref{lem:lossylink-simulation}, 
to prove that ensuring at most
$k$ different decision values is impossible here, as too many
decision values may originate from the unstable period.

\begin{lemma} \label{lem:lossylink-simulation}
Let $\M'=\li{p,q}$ be a two-processor subsystem of our system $\M=\li{\Pi}$.
If the sequence of communication graphs $\Gr$, $r>0$, of 
$\M$ are restricted by the existence of a round
$r'>0$ such that
(i) for $r < r'$, $(p \rightarrow q) \in \Gr$ and/or $(q \rightarrow p) \in \Gr$,
         and no other edges incident with $p$ or $q$ are in $\Gr$, and (ii)
for $r \geq r'$, there are no edges incident with $p$ and $q$ at all
         in $\Gr$,
then consensus is impossible in $\M'$.
\end{lemma}

\begin{proof}
Up to $r'$, this is ensured by the impossibility of 2-processor
consensus with a lossy but at least unidirectional link established in
\cite[Lemma 3]{SWK09}. After $r'$, this result continues to hold (and
is even ensured by the classic lossy link impossibility \cite{SW89}).
Hence, consensus is indeed impossible in $\M'$. \qed
\end{proof}

\begin{theorem} \label{thm:rGST-impossibility}
There is no algorithm that solves \ksa{} for $n \geq k+1$ processes under 
the message adversary \MAkd{\infty}, for every $1<k<n$.
\end{theorem}

\begin{proof}
Suppose again that there is a $k$-set algorithm $\A$ that works correctly under the
assumptions of our theorem. We restrict our attention to runs of $\M_\A$ where,
until $\rGST$, (i) the same set of $k-1$ root components
$\set{\D_1, \ldots, \D_{k-1}}$ with $\D = \bigcup_{i=1}^{k-1} \D_i$
exists in every round, and (ii) two remaining processes 
$\PiD = \Pi \setminus \D = \set{p_1, p_2}$ exist, which are 
(possibly only uni-directionally, i.e., via a lossy link) 
connected in every round, without additional edges to or from $\D$.
After $\rGST$, the communication graph remains the same, except that
the processes in $\PiD$ are disconnected from each other and there is
an edge from, say, $p_1$ to some process in $\D$ in every round. Note
that these runs satisfy \cref{ass:inter} for $d=\infty$, as the number
of root components never exceeds $k$.

Moreover, we let the adversary choose $\rGST$ sufficiently large such that the processes in $\D$ 
have decided. Since the processes in $\D_i$ ($i<0<k$) never receive 
a message from the remaining system before $\rGST$, in which case
they must eventually unilaterally decide, we can safely assume this.

We can now again employ the generic impossibility
\iftoggle{journal}{
\cref{thm:impossibility}
}{
from \cite[Theorem 1]{BRS11:OPODIS}
}
in this modified setting. The proofs of properties (A), (B) and
(D) remain essentially the same as in \cref{thm:D-static-roots}.
It hence only remains to prove:
\begin{compactdesc}
   \item[(C)] Consensus is impossible in $\mathbf{\M' = \left\langle \PiD
   \right\rangle }$: This follows immediately from
      \cref{lem:lossylink-simulation} with
      $r' = \rGST$. \qed
\end{compactdesc}
\end{proof}

The following \cref{thm:strong-rGST-impossibility} reveals that even
(considerably) less than $k$ root components per round before 
stabilization and a single perpetually stable root component
after stabilization are not sufficient for solving \ksa{}.

\newcommand\floor[1]{\left\lfloor #1 \right\rfloor}
\newcommand\ceil[1]{\left\lceil #1 \right\rceil}

\begin{theorem} \label{thm:strong-rGST-impossibility}
There is no algorithm that solves \ksa{} for $n \geq k+1$ processes under 
the message adversary \MAdv{\ceil{k/2} + 1}{\infty}, for every $1<k<n$, even if
$\G^{r} = \G$, $r\geq \rGST$, where $\G$ contains only a single root component.
\end{theorem}

\begin{proof}
We show that, under the assumption that $\A$ exists, there is a sequence of
communication graphs that is feasible for our message adversary that
leads to a contradiction. We choose $x_i = i$ for all $p_i \in \Pi$ and let
$\D_i = \set{p_{1+2i}, p_{2+2i}}$ for $0 \leq i < \ceil{k/2}-1$.
If $k$ is even, let $\D_{k/2-1}=\set{p_{k-1}, p_k}$;
if $k$ is odd, let $\D_{\ceil{k/2-1}}=\set{p_{k}}$.
In any case, let $\D_{\ceil{k/2}} = \{p_{k+1}\}$.
Finally, let $\PiD = \set{p_{k+2}, \ldots, p_n}$.
Note that $\PiD$ may be empty, while all $\D_i$ are guaranteed to contain at least
one process since $n>k$. For all rounds, the 
processes in $\PiD$ have an incoming edge from a process in one of the $\D_i$.

We split the description of the adversarial strategy into $\ceil{k/2}+1$ phases
in each of which we will force some $\D_i$ to take $|\D_i|$ decisions. 
To keep processes $p,q \in \D_i$ with $|\D_i|=2$ from deciding on the same value before 
their respective phase $i$, the adversary restricts $\Gr$ such that
(i) there are no links to $\D_i$ from any other $\D_j$ and (ii) 
either the edge $(p \rightarrow q)$ or $(p \leftarrow q)$ or both 
are in $\Gr$, in a way that causes \cref{lem:lossylink-simulation}
to apply. Note carefully that any such $\Gr$ indeed has no more than 
$\ceil{k/2} + 1$ root components.

In the initial phase, $\D_{\ceil{k/2}}$ is forced to decide: 
Since $p_{k+1}$ has no incoming edges from another node in $\Gr$,
this situation is indistinguishable from a run where $p_{k+1}$ became 
the single root after $\rGST$.
Thus, by the correctness of $\A$, $p_{k+1}$ must eventually decide on $x_{k+1}=k+1$.
At this point, the initial phase ends, and we can safely allow the 
adversary to modify $\Gr$ in such a way that $p_{k+1}$ has an incoming
edge from some other process.

We now proceed with $\ceil{k/2}-1$ phases: In the $i$\textsuperscript{th} phase,
$0 \leq i < \ceil{k/2}-1$, the adversary drops any link between 
the processes $p, q \in \D_i$ (and does not provide an incoming link from 
any other process, as before) in any $\Gr$.
Since, for both $p$ and $q$, this is again indistinguishable from the situation
where they become the single root after $\rGST$,
both will eventually decide in some future round (if they have not already decided).
Since the adversary may have chosen a link failure pattern in
earlier phases that causes the impossibility (= forever bivalent run) of 
\cref{lem:lossylink-simulation} to apply, as $\mathbf{\M'_{\A_{|\D_i}} \compat{\D_i} \M_\A}$, 
it follows that $\A$ and hence $\A_{|\D_i}$ cannot
have solved consensus in $\D_i$. Since $\A$ solves \ksa{}, $p$ and $q$ must hence
decide on two \emph{different}
values. Moreover, since neither $p$ nor $q$ ever received a message from a process not in $\D_i$,
their decision values must be different from the ones in all former phases.

Finally, after $p$ and $q$ have made their decisions, the adversary may again modify 
$\Gr$ such that they have an incoming edge from some other process, 
thereby reducing the number of root components by two and preserving the
maximum number $\ceil{k/2} + 1$  of root components, and continue with the next phase.

If $k$ is even, then the final phase $\ceil{k/2}-1$ 
forces two more decisions just 
as described above; otherwise, $p_k$ provides one additional decision value
(which happens concurrently with the initial 
phase here). In either case, we have shown that all $p_i$ with $1\le i\le k+1$ have 
decided on different values, which contradicts the assumption that a correct 
algorithm $\A$ exists. \qed
\end{proof}

Note that \cref{thm:strong-rGST-impossibility} reveals an interesting gap between
$2$-set agreement and 1-set agreement, i.e., consensus:
It shows that $2$-set agreement is impossible 
with $\lceil k/2\rceil+1=2$ root components per round before
and a single fixed root component after stabilization.
By contrast, if we reduce the number of root components per round
to a single one before stabilization (and still consider a single fixed root thereafter), 
even $1$-set agreement becomes solvable~\cite{BRS12:sirocco}.

\section{Algorithms for $k$-Set Agreement}
\label{sec:sufficiency}

In this section, we will provide a message adversary \MAJINF{k} 
(\cref{ass:stablestream}) that is sufficiently weak for solving \ksa{} 
if combined with \MAdv{n}{3D+\nwbound} (\cref{ass:inter}). 
Although we can of course not claim that it
is a strongest one in terms of problem solvability 
(we did not even define what this means), we
have some indications that it is close to the solvability/impossibility 
border.

\subsection{Set agreement}

To illustrate some of the ideas that will be used in our message adversary
for general \ksa{}, we start with the simple case of $n-1$-set 
agreement (also called \emph{set agreement}) first. Note that
\cref{thm:D-static-roots} does not apply here. To circumvent the
impossibility result of \cref{thm:strong-rGST-impossibility}, it suffices to
strengthen the assumption of at most $n-1$ root components
in every round such that the generation
of too many decision values during the unstable period is ruled out. 
A straightforward way to achieve this is to just forbid $n$ different
decisions obtained in root components consisting of a single
process. Achieving this is easy under the $\Sigma_{n-1}$-influence message
adversary given in \cref{asm:influenceset}, the name of which has
been inspired by the $\Sigma_{n-1}$ failure detector \cite{BR11:TCS}.

\begin{definition}[$\Sigma_{n-1}$-influence message adversary]\label{asm:influenceset}
The message adversary \MASigma\ is the set of all 
sequences of communication graphs $(\Gr)_{r>0}$, where in any set
$R_1^{I_1},\dots,R_n^{I_n}$ of $n$ root components consisting of
single processes $R_i=\{p_i\}$, $1\leq i \leq n$, occuring in any run
the following holds: There are two indices $i$, $j\neq i$ such 
that $R_i^{I_i}$ \emph{influences} $R_j^{I_j}$,
denoted $\influence{R_i^{I_i}}{R_j^{I_j}}$, in the sense that
there exists a causal chain starting after $I_i$ that ends before
or at the beginning of $I_j$.
\end{definition}

It is easy to devise a set agreement algorithm that works correctly
in a dynamic network under \cref{asm:influenceset}, provided
(a bound on) $n$ is known: In \cref{alg:setagreement}, process
$p_i$ maintains a proposal value $v_i$, initially $x_i$, and
a decision value $y_i$, initially $\bot$, which are broadcast in every 
round. If $p_i$ receives no message from any other process in a round,
it decides by setting $y_i=v_i$. If $p_i$ receives a message from some
$p_j$ that has already decided ($y_j\neq \bot$), it sets $y_i=y_j$. 
Otherwise, it updates $v_i$ to the maximum of $v_i$ and all 
received values $v_j$. At the end of round $n$, a process that has 
not yet decided sets $y_i:=v_i$, and all processes terminate.

\begin{algorithm}[h]

\begin{algorithmic}[1]
\scriptsize
\setcounter{ALC@unique}{0}
\setlinenosize{\scriptsize}
\setlinenofont{\tt}
\item[] {\bf Set agreement algorithm, code for process $p_i$:}
\STATE $v_i := x_i \in V$ // initial value
\STATE $y_i:=\bot$
\item[] {\bf\boldmath Emit round $r$ messages:}
\STATE send $\li{v_i,y_i}$ to all
\item[] {\bf\boldmath Receive round $r$ messages:}
\STATE receive $\li{v_j,y_j}$ from all current neighbors
\item[] {\bf\boldmath Round $r$: computation:}
	\STATE $v_i := \max\{v_i, v_j: j \in \N_{p_i}\}$
	\IF{$\exists j: (y_j \neq \bot) \wedge (y_i= \bot)$}
		\STATE $y_i:=y_j$
        \ENDIF
	\IF{$(\N_{p_i}=\emptyset) \wedge (y_i= \bot)$}
		\STATE $y_i:=v_i$
        \ENDIF
	\IF{$(r=n) \wedge (y_i= \bot)$}
		\STATE $y_i:=v_i$; terminate
        \ENDIF
\end{algorithmic}%

\caption{Set agreement algorithm for message adversary \MASigma.}\label{alg:setagreement}

\end{algorithm}

\begin{theorem}[Correctness \cref{alg:setagreement}] \label{thm:setagreement}
\cref{alg:setagreement} solves $n-1$-set agreement in a
dynamic network under message adversary \MASigma\ given in \cref{asm:influenceset}.
\end{theorem}

\begin{proof}
Termination (after $n$ rounds) and also validity are obvious, 
so it only remains to show $n-1$-agreement. 
Assume, w.l.o.g., that the processes $p_1,p_2,\dots$ are ordered according to
their initial values $x_1\leq x_2 \leq \dots$, and let $S^k$ be the 
set of different values (in $y_i$ or, if still $y_i=\bot$, in $v_i$) 
present in the system at the 
beginning of round $k\geq 1$; $S^1=\{x_1,\dots,x_n\}$ is the set of 
initial values. Obviously, $S^1 \supseteq S^2 \supseteq \dots$, and
since $n-1$-agreement is fulfilled if $|S^{n+1}|<n$, we only need to
consider the case where all $x_i$ are different.

Consider process $p_1$: If $p_1$ gets a message from some other process 
$p_j$ in round 1, $x_1 \not\in S^2$ as (i) $p_1$ does not decide on its
own value and sets $v_1\geq v_j \geq x_j > x_1$ and (ii) no process that 
receives a message containing $x_1$ from $p_1$ takes on this value.
Hence, $n-1$-set agreement will be achieved in this case. Otherwise, $p_1$ does not
get any message in round 1 and hence decides on $x_1$.

Proceeding inductively, assume that $p_\ell \in P^{i-1}=\{p_1,\dots,p_{i-1}\}$ 
has decided on $x_\ell$ by round $k\leq\ell$, and received only messages
from processes with smaller index in rounds $1,\dots,k-1$ and
no message in round $k$. Now consider
process $p_i$: If $p_i$ gets a message from some process $p_j$ with $j>i$
in some round $k\leq i$, with minimal $k$, before it decides, then $x_i \not\in S^{k+1}$ as 
(i) $p_i$ does not decide on its own value and sets $v_i\geq v_j \geq x_j > x_i$,
(ii) $p_i$ did not send its value to any process in $P^{i-1}$ before their
decisions, and (iii) no process with index larger than $i$ that receives a 
message containing $x_i$ from $p_i$ takes on this value. 
Hence, $n-1$-set agreement will be achieved in this case. Otherwise, if $p_i$ gets a
message from some process $p_\ell\in P^{i-1}$ in round $i$, it will 
decide on $p_\ell$'s decision value $x_\ell$ and hence also cause 
$x_i \not\in S^{i+1}$. In the only remaining case, $p_i$ does not
get any message in round $i$ and hence decides on $x_i$, which completes
the inductive construction of $P^{i}=\{p_1,\dots,p_{i}\}$ for $i<n$.

Now consider $p_n$ in round $n$ in the above construction of $P^n$: 
\cref{asm:influenceset} prohibits
the only case where $n-1$-agreement could possibly be violated, namely, when $p_n$
also decides on $x_n$: During the first $n$ rounds, we would have obtained
$n$ single-node root components no two of which influence each
other in this case. 
Thus, we cannot extend the inductive construction of $P^i$ to $i=n$,
as the resulting execution would be infeasible. \qed
\end{proof}

\subsection{A message adversary for general \ksa{}}

Whereas the set agreement solution introduced in the previous subsection
is simple, it is apparent that \cref{asm:influenceset}
is quite demanding. In particular, it requires explicit knowledge of 
(a bound on) $n$. 
We will now provide a message adversary \MAJINF{k} (\cref{ass:stablestream}), 
which is sufficient for general $k$-set agreement if combined with
\MAdv{n}{3D+\nwbound} (\cref{ass:inter}). We obtained this combination
by adding some additional properties to the necessary
network conditions implied by our impossibility
\cref{thm:D-static-roots,thm:strong-rGST-impossibility}.\footnote{An alternative 
way to derive sufficient network assumptions
for, e.g., $n-2$-set agreement could be to generalize \cref{asm:influenceset}: One
could e.g.\ assume that at least two out of every set of $n-1$ different root components 
consisting of 1 or 2 processes are influenced by a common predecessor root component. 
Whereas this assumption does not require vertex stability of root components, it
effectively ensures that information propagates not slower as in VSRCs. Owing to this fact, 
it also prohibits the existence of the node $q$ in \cref{ass:too-short} with causal 
distance $D$ from $p$ in the root component, thereby causing the proof of 
\cref{thm:D-static-roots} to fail. Working out the details may turn out difficult,
though: After all, unlike single-process roots, larger root components suffer from
the problem that its members cannot always determine whether the root was a VSRC
or not. Influence must hence be conservative, in the sense that it involves
even \emph{potential} 2-process roots.}

To avoid non-terminating (i.e., forever undecided) executions
as predicted by \cref{thm:D-static-roots}, we require the \emph{stable 
interval} constraint guaranteed by the message adversary \MAdv{n}{3D+\nwbound} 
to hold.
The parameter $D$, which can always be safely set to $D=n-1$ 
according to \cref{lem:infprop}, 
allows to adapt the message adversary to the actual dynamic causal diameter 
guaranteed in the VSRCs of a given dynamic network. Note that, since $D>0$, rounds
where no message is received are not forbidden here (in contrast to 
\cref{asm:influenceset}).

In order to also circumvent executions violating
the $k$-agreement property established by \cref{thm:strong-rGST-impossibility},
we introduce the \emph{majority
 influence} constraint guaranteed by the message adversary \MAJINF{k}
given in \cref{ass:stablestream} below.
Like \cref{asm:influenceset} for set agreement,
it guarantees some (minimal) information
flow between sufficiently long-lasting vertex-stable root
components that exist at different
times. As visualized in \cref{fig:model_majorinf}, it implies that the information
available in any such VSRC originates in at most $k$ ``initial'' VSRCs. 
Thereby, it enhances the very limited information propagation
that could occur in our model solely under \MAkd{3D+\nwbound}, which
is too strong for solving $k$-agreement. 

\begin{figure}

  \begin{minipage}[c]{0.5\textwidth}
\centering
\begin{tikzpicture}[->,>=stealth',shorten >=1pt,auto,node distance=2.8cm,
                    semithick,scale=0.25]
  \tikzstyle{state}=[circle,draw,text=black,inner sep=3mm,scale=0.6]
  \tikzstyle{mystate}=[circle,fill=black!20,draw=black,text=black,inner sep=3mm,scale=0.6]
  \tikzstyle{myarrows}=[line width=0.8mm,draw=black]

  \node[state]	       (A)                 {$R_1^{I_1}$};
  \node[mystate,fill=black!40]         (B) [right of=A]       {$R_2^{I_2}$};
  \node[mystate]         (C) [right of=B]       {$R_3^{I_3}$};
  \node[mystate]         (D) [right of=C]       {$R_4^{I_4}$};

  \node[state]         (I) [below left of=B]       {$R_5^{I_5}$};
  \node[state]         (K) [below right of=B]       {$R_6^{I_6}$};

  \node[mystate,fill=black!40]         (F) [below left of=K]       {$R_8^{I_8}$};
  \node[mystate]         (G) [right of=F]       {$R_9^{I_9}$};
  \node[mystate]         (H) [right of=G]       {$R_{10}^{I_{10}}$};
  \node[state]	       (E) [left of=F,]      {$R_{7}^{I_7}$};

  \path (A) edge              node {} (B)
        (B) edge[myarrows]    node {} (C)
        (C) edge[myarrows]    node {} (D)

        (I) edge              node {} (F)
        (I) edge              node {} (B)
        (K) edge              node {} (D)

        (E) edge              node {} (F)
        (F) edge[myarrows]    node {} (G)
	(F) edge              node {} (K)
        (G) edge[myarrows]    node {} (H);

\draw[->,draw=black,line width=0.3mm] ($ (E) + (-2,12) $) -- ($ (E) + (23,12) $)node[below right] {$t$};

\end{tikzpicture}

  \end{minipage}
\hspace*{1cm}%
 \begin{minipage}[c]{0.4\textwidth}
\centering
\caption{\small
VSRCs influencing each other in a run, for $k=2$. Time progresses from left
to right; all \textcolor{black!100}{shaded} nodes are stable for more than $2D$
rounds, white nodes are stable between $D+1$ and $2D$ rounds. Thick arrows
represent majority influence, thin arrows represent (weak) influence. 
At most two \textcolor{black!100}{shaded}
nodes, depicted \textcolor{black!100}{darkly shaded}, 
may exist that are not majority-influenced by another shaded node.}
\label{fig:model_majorinf}
  \end{minipage}
\end{figure}
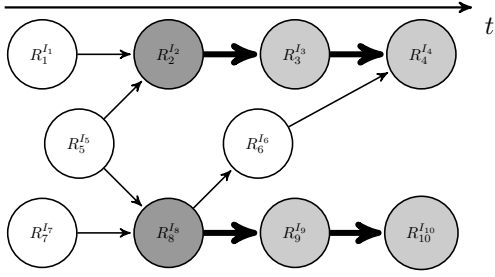

Formally, given some run $\rho$, we denote by $\mathbb{V}_d$ the set of all root components
that are vertex-stable for at least $d$ consecutive rounds in $\rho$.
Let $R_{cur} \in \mathbb{V}_1$ be vertex-stable in $I_{cur}=[r_{cur},s_{cur}]$
and $R_{suc} \in \mathbb{V}_1$ be vertex-stable in $I_{suc}=[r_{suc},s_{suc}]$
with $r_{suc} > s_{cur}$; note that $\mathbb{V}_d \subseteq \mathbb{V}_1$ for every $d\geq 1$.

\begin{definition}[(Weak) Influence]\label{def:influence}
Given any two $R_{cur}^{I_{cur}}, R_{suc}^{I_{suc}} \in \mathbb{V}_1$, we say that
some process $p \in R_{cur}^{I_{cur}}$ \emph{influences}
some process $q \in R_{suc}^{I_{suc}}$ and write
$\influence{p}{q}$ with $\influence{}{} \subseteq \Pi^2$ iff
there exists a causal chain from $p$ to $q$ starting after $I_{cur}$ that ends before
or at the beginning of $I_{suc}$, i.e.,
$cd^{s_{cur}+1}(p, q) \leq r_{suc} - s_{cur}$. In this case,
we also say that $R_{cur}^{I_{cur}}$ (weakly) influences $R_{suc}^{I_{suc}}$ and write
$\influence{R_{cur}^{I_{cur}}}{R_{suc}^{I_{suc}}}$, using the relation
$\influence{}{} \subseteq \mathbb{V}_1^2$ here.
\end{definition}

We will also need stronger notions of influence, which are based on the following
\cref{def:influence-set}:

\newcommand{\IS}{\mbox{IS}}
\begin{definition}[Influence Sets]\label{def:influence-set}
Given any two $R_{cur}, R_{suc} \in \mathbb{V}_1$, their influence set is
$\IS(R_{cur},R_{suc}):=\set{q \in R_{suc} \mid \exists p \in R_{cur} \colon \influence{p}{q}}$.
\end{definition}

The \emph{majority influence} between the nodes in $R_{cur}$ and $R_{suc}$ guarantees that
 $R_{cur}$ influences a set of nodes in $R_{suc}$, which is greater than any
set influenced by VSRCs not already known by the processes in $R_{cur}$
(and greater than or equal to any set influenced by VSRCs already known by the
processes in $R_{cur}$). Majority influence is hence a very natural way to 
discriminate between strong and weak influence between VSRCs, see
\cref{def:stronginfluence} below.

\begin{definition}[Majority influence]\label{def:majorityInfluence}
We say that a VSRC
$R_{cur} \in \mathbb{V}_{2D+1}$ exercises a \emph{majority influence} on a
VSRC $R_{suc} \in \mathbb{V}_{2D+1}$,
denoted $\majinf{R_{cur}}{R_{suc}}$ with
$\majinf{}{} \subseteq \mathbb{V}_{2D+1}^2$, iff
$\forall R \in \mathbb{V}_{D+1}$ with
$\IS(R, R_{cur}) = \emptyset$ it holds that
$|\IS(R_{cur}, R_{suc})| > |\IS(R, R_{suc})|$ and
$\forall R \in \mathbb{V}_{D+1}$ with $\IS(R, R_{cur}) \neq \emptyset$ it holds
that
$|\IS(R_{cur}, R_{suc})| \geq |\IS(R, R_{suc})|$.
\end{definition}

\iftoggle{journal}{
The relation $\majinf{}{}$ has the following properties:

\begin{lemma}[Properties $\majinf{}{}$]\label{lem:propmaj}
 The majority influence relation is
antisymmetric, acyclic and intransitive.
\end{lemma}
\begin{pf}
Let $R$, $\overline{R}$, and $\hat{R}$ be three different VSRCs stable in the intervals
$I$, $\overline{I}$, and $\hat{I}$, resp.
Since the VSRCs $R$ and $\overline{R}\neq R$
are ordered in time according to their round intervals $I$ and $\overline{I}$,
which must be disjoint, no process in $R$ can be influenced by any process in
$\overline{R}$ if $\majinf{R}{\overline{R}}$. 
Hence, $\majinf{\overline{R}}{R}$ cannot hold,
which implies both antisymmetry and, by a transitive application of
this argument, acyclicity.
To prove intransitivity, observe that $\majinf{R}{\overline{R}}$
and $\majinf{\overline{R}}{\hat{R}}$ would imply
$\IS(R,\hat{R}) > \IS(\overline{R},\hat{R})$ if 
$\majinf{R}{\hat{R}}$ also held,
since no process in $R$ can be influenced by any process 
in $\overline{R}$. This contradicts $\IS(\overline{R},\hat{R}) \geq \IS(R,\hat{R})$
required by $\majinf{\overline{R}}{\hat{R}}$,
however.
\end{pf}

\begin{definition}[Strong Influence]\label{def:stronginfluence}
We say that $R_{cur} \in \mathbb{V}_{2D+1}$ \emph{strongly influences} 
$R_{suc} \in \mathbb{V}_{2D+1}$ and write
$\stronginf{R_{cur}}{R_{suc}}$,
where $\stronginf{}{}\subseteq \mathbb{V}_{2D+1}^2$ is the transitive closure of $\majinf{}{}$.
\end{definition}
Note carefully that $\stronginf{}{}$ is antisymmetric by \cref{lem:propmaj}.
}{}

With these preparations, we are now ready to specify a message adversary 
\MAJINF{k} given in \cref{ass:stablestream}. 

\begin{definition}[$k$-majority influence message adversary] \label{ass:stablestream}
The message adversary \MAJINF{k} is the set of all 
sequences of communication graphs $(\Gr)_{r>0}$, where in every run
$\exists K \subseteq \mathbb{V}_{2D+1}$ with $|K| \leq k $ s.t.\
$\forall \overline{R} \in \mathbb{V}_{2D+1}\setminus K ~ \exists R \in \mathbb{V}_{2D+1}$ with $\majinf{R}{\overline{R}}$. 
\end{definition}

Informally speaking, \cref{ass:stablestream} ensures that all but at most $k$
``initial'' VSRCs in $\mathbb{V}_{2D+1}$ are majority-influenced by some earlier VSRC in $\mathbb{V}_{2D+1}$ (see \cref{fig:model_majorinf}). Note carefully, though, that \cref{ass:stablestream}
neither prohibits partitioning of the system in more than $k$ simultaneous VSRCs
nor directly exhibits a $k$-quorum property, cf.\ the well-known quorum failure detector $\Sigma_k$ \cite{BR11:TCS}
that is known to be necessary (but not sufficient!) for solving $k$-set agreement: 
After all, one could e.g.\ choose $k+1=3$
VSRC's $R_2^{I_2}$, $R_4^{I_4}$ and $R_7^{I_7}$ in \cref{fig:model_majorinf} without
finding any pair among those which are majority-influenced by a common predecessor VSRC.
Therefore, \MAJINF{k} alone is too strong for solving \ksa{}. The same is true 
for an alternative to \cref{ass:stablestream} that just ensures a $k$-quorum
(unless acyclicity could be guaranteed as well).

Conversely, if majority influence was replaced by strong influence according to 
\cref{def:stronginfluence}, a quorum property could be easily established: Starting
out from an arbitrary set of $k+1$ $2D+1$-VSRCs, we could go back along the (acyclic)
majority influence relation until we end up in the set $K$ guaranteed by 
\cref{ass:stablestream}. If a $k$-set agreement algorithm 
relied on $2D+1$-VSRCs for decisions, this would guarantee that no more than
$k$ decision values (possibly fabricated in the ``initial'' $2D+1$-VSRCs) can be produced. 
A message adversary equivalent to \cref{ass:stablestream} with strong majority
would be fairly weak, however.

These observations indicate that \MAdvI{k}{3D+\nwbound} is indeed reasonably close 
to the $k$-set agreement solvability border. 

\smallskip

We conclude this section with some straightforward stronger assumptions,
which also imply \cref{ass:stablestream} and can hence be handled by
the algorithm introduced in \cref{sec:algorithms}:
\begin{compactenum}
\item[(i)] Replacing majority influence %
in \cref{def:majorityInfluence} by majority intersection
$\lvert R_{suc} \cap R \rvert < \lvert R_{suc} \cap R_{cur} \rvert$,
which is obviously the strongest form of influence.
\item[(ii)] Requiring $\lvert R_{suc} \cap R_{cur} \rvert >
\lvert R_{suc} \rvert/2$, i.e., a majority intersection
with respect to the number of processes in $R_{suc}$.
This could be interpreted as a changing VSRC, in the
sense of ``$R_{suc}$ is the result
of changing a minority of processes in $R_{cur}$''. Although
this restricts the rate of growth of VSRCs in a run,
it would apply, for example, in case of random graphs where the
giant component has formed \cite{DMS01,JKLP93}.
\end{compactenum}

\subsection{Gracefully degrading consensus/\ksa{}}
\label{sec:algorithms}

In this section,
we provide a $k$-set agreement algorithm and prove that
it works correctly under the message adversary \MAkI, i.e., the conjunction of 
\cref{ass:inter,ass:stablestream}.
Note that the algorithm needs to know $D$, but neither $n$ nor $\nwbound$. It
consists of a ``generic'' $k$-set agreement algorithm, which
relies on the network approximation algorithm of \cref{sec:approxalgo} for locally
detecting vertex-stable root components and a function $\lambdaabstraction$
that extracts candidate decision values from history
information. Our implementation of $\lambdaabstraction$ uses a vector-clock-like
mechanism for maintaining ``causally consistent'' history information, which
can be guaranteed to lead to proper candidate values thanks to \MAkI.

In sharp contrast to classic $k$-set agreement algorithms, the algorithm is
\emph{$k$-uniform}, i.e., the parameter $k$ does not appear in its code.
Rather, the number of system-wide decision values is determined by the number of
(certain) $2D+1$-VSRCs occurring in the particular run.
As a consequence, if the network partitions into $k$ weakly connected components, 
for example,\footnote{It 
is important to note that the network properties
required by our algorithm to reach $k$ decision values
need \emph{not} involve $k$ isolated partitions: Obviously, $k$ isolated partitions
in the communication graph also imply $k$ root components, but $k$ root
components do not imply a partitioning of the communication graph
into $k$ weakly connected components --- one process may still be 
connected to several components.}
all processes in a component obtain the same
decision value. On the other hand,
if the network remains well-connected, the algorithm
guarantees a unique decision value system-wide. 

Our algorithm is in fact not only $k$-uniform but even worst-case $k$-optimal,
in the sense that (i) it provides at most $k$ decisions system-wide in all
runs that are feasible for \MAkI, and (ii) that there is at least one feasible 
run under \MAkI\ where no correct \ksa{} can guarantee less than $k$ decisions.
(i) will be proved in \cref{sec:corproof}, and (ii) follows immediately from
the fact that a run consisting of $k$ isolated partitions is also feasible
for \MAkI. Our algorithm can hence indeed be viewed as a
consensus algorithm that degrades gracefully to $k$-set agreement, for some $k$
determined by the actual network properties.

\def\lock{\ensuremath{\mbox{\tt lock}}}
\def\term{\ensuremath{\mbox{\tt decide}}}
\def\r0{1}
\def\r-1{0}
\def\rs{rs}
\def\lockround{\ensuremath{\ell}}
\def\maxID{\ensuremath{\texttt{maxID}}}
\def\NrS{\ensuremath{\#_{\texttt{Supt}}}} %
\def\SS{\ensuremath{\texttt{Supt}}} %
\def\redate{\ensuremath{\texttt{redate}}}
\def\last{\ensuremath{\texttt{ last}}}
\def\maj{\ensuremath{\mbox{mfrq}}}
\def\myRoot{\texttt{myRoot}}
\def\rStart{\ensuremath{r_{\mbox{\textit{start}}}}}
\def\hot{\tau_{\mathtt{create}}}
\def\newLock{\mathtt{newLock}}
\def\hotmaj{\ensuremath{\maj_{\mathtt{latest}}}}
\def\decision{\mathtt{decision}}

\medskip

Like the consensus algorithm in \cref{sec:consensus},
our \ksa{} algorithm consists of two reasonably independent parts, the 
network approximation algorithm \cref{alg:approx}
and the $k$-set agreement core algorithm given in \cref{alg:ksa}.
As in \cref{sec:consensuscore}, we assume that the complete round $r$ computing step
of the network approximation algorithm is executed just
before the round $r$ computing step of the $k$-set algorithm,
and that the round~$r$ message of the former is piggybacked
on the round~$r$ message of the latter.
Recall that this implies that the round $r$ computing
step of the $k$-set core algorithm, which terminates round~$r$,
can already access the result
of the round $r$ computation of the network approximation algorithm,
i.e., its state at the end of round $r$.

\begin{algorithm}[h]
\begin{algorithmic}[1]
\setcounter{ALC@unique}{0}
\scriptsize
\setlinenosize{\scriptsize}
\setlinenofont{\tt}
\item[] {\textbf{Variables and Initialization:}}
\STATE $\hist_i[*][*] := \emptyset$ /* $\hist_i[j][r]$ holds $p_i$'s estimate 
of the locks learned by $p_j$ in round $r$ */
\STATE $\hist_i[i][0] := \set{(\set{p_i}, x_i, \r-1)}$ /* virtual first lock 
$(V(R):=\set{p_i}, v:=x_i,
\hot:=0)$ at $p_i$ */
\label{line:initHist}
\STATE $\lockround := \bot$ // most recent lock round, $\bot$ if none
\STATE $\decision_i := \bot$ // $p_i$'s decision, $\bot$ if undecided

\item[] {\textbf{ Emit round r messages:}}
\STATE send $\li{\hist_i, \decision_i}$ to all neighbors \label{line:sendHistory}

\item[] {\textbf{ Receive round r messages:}}
\STATE for all $p_j$ in  $p_i$'s neighborhood $\N^r_{p_i}$, receive
$\li{\hist_j, \decision_j}$

\item[] {\textbf{ Round r computation:}}
\IF {$\decision_i = \bot$}
\IF {received any message $m$ containing $m.\decision \neq \bot$}
  \STATE decide $m.\decision$ and set $\decision_i := m.\decision$ \label{line:adoptdec}
\ELSE
\item[] // update $\hist_i$ with $\hist_j$ received  from neighbors
\FOR{$p_j \in \N^r_{p_i}$, where $p_j$ sent $\hist_j$}
  \STATE $\hist'_i := \hist_i$ // remember current history
  \FOR{all non-empty entries $\hist_j[x][r']$ of $\hist_j$, $x \neq i$}
    \STATE $\hist_i[x][r'] := \hist_i[x][r'] \cup \hist_j[x][r']$ \label{line:remoteHistoryUpdate}
  \ENDFOR
  \item[] // locally add all newly learned locks:
  \STATE $\hist_i[i] := \hist_i \setminus \hist'_i$ \label{line:learnLock}
\ENDFOR

\item[] // perform state transitions (undecided, locked, decided):
\STATE $\myRoot := \stableSCC(r-2D, r-D)$
\IF {$\lockround = \bot$ and $\myRoot \neq \emptyset$} \label{line:lock?}
    \STATE $\lockround := r-2D$
    \STATE $\lock := \getred(\myRoot, \lockround)$ \label{line:setLock}
    \STATE $\hist_i[i][r] := \hist_i[i][r] \cup \lock$ // create new lock
             \label{line:addNewLock}
\ELSIF {$\lockround \neq \bot$ and $\myRoot = \emptyset$}
    \STATE $\lockround := \bot$ // release unsuccessful lock \label{line:releaseLock}
\ELSIF {$\lockround \neq \bot$ and
  $\stableSCC{[\lockround, \lockround +2D] \neq \emptyset}$}
  \label{line:decide?}
    \STATE decide $\lock.v$ and set $\decision_i := \lock.v$ \label{line:setDecision}
\ENDIF
\ENDIF
\ENDIF

\EMPTY
\FUNC{$\getred(R, r')$}
  \STATE Let $S$ be the multiset $\bigcup_{p_j \in R, r'' \le r'} \hist_i[j][r'']$ \\
         Let $\maj(S)$ be the set of the most frequent elements in $S$
         \label{line:bounded-history}
  \STATE Let
      $\hotmaj(S) := \set{x \in \maj(S) \mid \forall y \neq x \in \maj(S) \colon x.\hot > y.\hot}$
  \IF{$|\hotmaj(S)| = 1$} \label{line:hottestMajority?}
      \STATE Let $v$ be $s.v$ of the single element $s \in \hotmaj(S)$
      \STATE $\newLock := (R, v, r)$
          \label{line:createLockMajority}
  \ELSE \label{line:ambiguousMajority?}
     \STATE $\newLock := (R, \max_{s \in S}\set{s.v}, r)$ \label{line:createLockMaximum} // deterministic choice
  \ENDIF
  \RETURN $\newLock$ \label{line:returnNewLock}
\ENDFUNC

\end{algorithmic}

\caption{$k$-uniform $k$-set agreement algorithm, code for process $p_i$}\label{alg:ksa}

\end{algorithm}

The general idea of our core $k$-set agreement algorithm in
\cref{alg:ksa} is to generate new decision values
only at members of $2D+1$-VSRCs, and to disseminate those values throughout
the remaining network. 
Using the network approximation $A_{p_i}$, our algorithm causes process $p_i$
to make a transition from the initially \emph{undecided} state to a \emph{locked} state
when it detects some minimal ``stability of its surroundings'', namely, its
membership in some $D+1$-VSRC $D$ rounds in the past (\cref{line:lock?}).
Note that the latency of $D$ rounds is inevitable here, since information propagation
within a $D+1$-VSRC may take up to $D$ rounds due to $D$-boundedness, as 
guaranteed by item~(ii) in \cref{ass:inter}.
If process $p_i$, while in the locked state, observes some period of stability
that is sufficient for locally inferring a consistent view among \emph{all}
VSRC members (which occurs when the $D+1$-VSRC has
actually extended to a $2D+1$-VSRC), $p_i$ can safely make a transition to the
\emph{decided} state (\cref{line:setDecision}). The decision value is
then broadcast in all subsequent rounds, and adopted by any not-yet decided
process in the system that receives it later on (\cref{line:adoptdec}).
Note that \MAdv{n}{3D+\nwbound} (\cref{ass:inter}) guarantees that
this will eventually happen.

Since locking is done optimistically, however, it may also happen that the $D+1$-VSRC
does not extend to a $2D+1$-VSRC (or, even worse, is not recognized to have done so by
some members) later on. In this case, $p_i$ makes a transition
from the locked state back to the undecided state (\cref{line:releaseLock}).
Unfortunately, this possibility has severe consequences:
Meachanisms are required that, despite possibly inconsistently perceived unsuccessful
locks, ensure both (a) an \emph{identical} decision value among all members of a $2D+1$-VSRC
who successfully detect this $2D+1$-VSRC and thus reach the decided state, and (b) no
more than $k$ different decision values originating from different $2D+1$-VSRCs.

Both goals are accomplished by a particular selection of the
decision values (using function $\getred$), which ultimately relies on an intricate utilization the network properties guaranteed by our message adversary \MAkI (\cref{ass:inter,ass:stablestream}): Our algorithm uses a suitable
\emph{lock
history} data structure for this purpose, which is continuously exchanged and updated among
all reachable processes. It is used to store sets of \emph{locks} $L=(R, v, \hot)$,
which are created by every process that enters the locked state: $R$ is the vertex-set of the
detected $D+1$-VSRC, $v$ is a certain proposal value (determined as explained below),
and $\hot$ is the round when the lock is created.

In more detail, the lock history at process $p_i$ consists of an array $\hist_i[j][r]$ that holds $p_i$'s (under)approximation of the locks process $p_j$ got to know in round $r$.
It is maintained using the following simple update rules:
\begin{compactitem}
\item[(i)] \emph{Local lock creation:} Apart from the single \emph{virtual} lock
$(\set{p_i}, x_i, 0)$ created initially by $p_i$ in \cref{line:initHist} (which
guarantees a non-empty lock history right from the beginning),
all regular locks created upon $p_i$'s transition from the undecided to the locked
state are computed by the function $\getred$ in \cref{line:setLock}.
Any lock locally created at $p_i$ in round $r$ (that is, in the round~$r$ computing step of
the core $k$-set agreement algorithm that terminates round $r$) is of course put
into $\hist_i[i][r]$.

\item[(ii)] \emph{Remote lock learning:} Since all processes exchange their lock histories,
$p_i$ may learn about some lock $L$ created by process $p_x$ in round $r'$ from
the lock history $\hist_j[x][r']$ received from some $p_j$ later on. In this case, $L$ is
just added to $\hist_i[x][r']$ (\cref{line:remoteHistoryUpdate}).

\item[(iii)] \emph{Local lock learning:} In order to ensure that the lock histories
of all members of a $2D+1$-VSRC are eventually consistent, which will finally ensure
identical decision values, \emph{every} newly
learned remote lock $L \in \hist_i[x][r']$ obtained in (ii) is also added to
$\hist_i[i][r]$.
\end{compactitem}
Note that the update rules (i)+(ii) resemble the ones of vector clocks \cite{Mat89}.

Clearly, $\hist_i[i][r']$ will always be accurate for current and past rounds
$r'\leq r$, while $\hist_i[j][r']$ may not always be up-to date, i.e., may lack some
locks that are present in $\hist_j[j][r']$.
Nevertheless, if $p_i$ and $p_j$ are members of the same $2D+1$-VSRC $R^I$ with $I=[r-2D,r]$,
\cref{def:D-bounded-VSRC} ensures that $p_i$ and $p_j$ have consistent histories
$\hist_i[j][r']$ and $\hist_j[i][r']$ at latest by (the end of) round $r'+D$,
for any $r' \in [r-2D, r-D]$.
Hence, if $p_i$ creates a new lock $L$ when it detects, in its round $r$ computing step,
that it was part of a $D+1$-VSRC that was stable from $r-2D$ to $r-D$, it is
ascertained that any other member $p_j$ will have locally learned
the same lock $L$ in the same round $r$, provided that the $D+1$-VSRC in fact extended
to a $2D+1$-VSRC.

The resulting consistency of the histories is finally exploited by the function $\getred(R,\ell)$, which
computes (the value of) a new local lock (\cref{line:setLock}) created in round $r$. As its
input parameters, it is provided with the members $R$ of the detected $D+1$-VSRC and its starting
round $\ell=r-2D$.
$\getred$ first determines a multiset $S$, which contains all locks locally known to the
members $p_j \in R$ by round $r-2D$ (\cref{line:bounded-history}). Note that the multiplicity
of some lock $L=(R',v,r')$ in $S$ is just the number of members of $R$ who got to know $L$
by round $r-2D$, which is just $|\IS(R',R)|$ according to \cref{def:influence-set}.
In order to determine a proper value for the new lock to be computed by $\getred$, we
exploit the fact that \MAJINF{k} (given in \cref{ass:stablestream}) ensures 
majority influence according to \cref{def:majorityInfluence}:
If the set $\hotmaj(S)$, containing the most frequent locks in $S$ with the same
maximal lock creation round, contains a single lock $L$ only, its value $L.v$
is used. Note that the restriction to the maximal lock creation date automatically filters unwanted,
outdated locks that have merely been disseminated in preceding $2D+1$-VSRCs, see (1) below.
Otherwise, i.e., if $\hotmaj(S)$ contains multiple candidate locks,
a consistent deterministic choice, namely, the maximum
among all lock values in $S$, is used (\cref{line:createLockMaximum}). As a consequence,
at most $k$ different decision values will be generated system-wide.

\smallskip

Given the various mechanisms employed in our algorithm and their complex interplay,
the question about a more light-weight alternative solution that omits some of these
mechanisms might arise. We will proceed with some informal arguments that support
the necessity some of the pillars of our solution, namely, (1) the preference of most recently
created locks in $\getred$, (2) the creation of a new lock at every transition to the
locked state, and finally (3) the usage of an a priori unbounded
data structure $\hist_i$.
Although these arguments are also ``embedded'' in the correctness proof
in the following section, they do not immediately leap to the eye and are
hence provided explicitly here.
\begin{compactitem}
  \item[(1)] The preference of most recently created lock in $\getred$, which is
done by selecting the set $\hotmaj(S)$ in \cref{line:hottestMajority?}, defeats
the inevitable ``amplification'' of the number of processes that got to know some
``old'' lock: All members of a $2D+1$-VSRC have finally learned \emph{all} ``old''
locks that were only known to \emph{some} of its members at the starting round of
the VSRC initially. In terms of multiplicity in $S$, this would falsely make any
such old lock a preferable alternative to the most recently created lock.

 \item[(2)] Instead of creating new locks at every newly detected $D+1$-VSRC,
       it might seem sufficient to simply update
       the creation time of an old lock that (dominantly) influences a newly detected
VSRC. This is not the case, however:
       Consider a hypothesized algorithm where new locks are only generated if no
       suitable old
       locks can be found in the current history, and assume
       a run where two VSRCs with vertex sets
       $R_1 = \set{p_1,p_2}$ and
       $R_3 = \set{p_1,p_2}$ that
       are both stable for $D+1$ rounds and two root components $R_2 = \set{p_1,p_3}$
       and $R_4 = \set{p_1,p_3}$ that are stable for $2D+1$ rounds are formed.
       Let these VSRCs be such that $R_i$ is formed before $R_j$ if $i<j$ and
       let there be no influence among the processes of $\set{p_1,p_2,p_3}$, apart
       from their influence on each other when they are members of the same VSRC.
       First, let the processes of $R_1$ lock on some old lock $L'$.
       Then, assume that the processes of $R_2$ lock
       on some lock\footnote{This could occur, e.g., because $L$ is known to
       $p_3$ and has a more
       recent creation time than $L'$} $L \neq L'$,
       a lock not known in $R_1$.
       Since $R_3 = \set{p_1,p_2}$, if $R_3$ is sufficiently well connected,
       $p_1$ might lock on
       $L'$ in $R_3$, because $L'$ is known to both $p_1$ and $p_2$ while $L$ is
       known merely to $p_1$ at the start of $R_3$.
       Subsequently, this results in the situation in $R_4$ where there is
       neither a clear majority ($L'$ and $L$ are known to both members of $R_4$)
       nor a clear most recently adopted lock (for $p_1$, it seems that $L'$ is the most
       recent lock, while for $p_3$, it seems that $L$ is more recent).
       Consequently, in $R_4$, it is not clear whether to lock on $L.v$ or on $L'.v$.
       Nevertheless, the processes of $R_4$ should be able to determine that they
       must lock on $L$ and not on $L'$, since $\majinf{R_2}{R_4}$ holds in our
       example: $|\IS(R_1,R_2)| = 1, |\IS(R_1,R_4)|=2, |\IS(R_2,R_4)|=2$ and
       $|\IS(R_3, R_4)| = 1$.
       We can therefore conclude that merely adopting old locks is insufficient.

 \item[(3)] Since the stabilization round $\rGST$, as implied by
        \cref{ass:inter}, may be delayed arbitrarily, an unbounded
number of $2D+1$-VSRCs can occur before $\rGST$. Since any of those
might produce a critical lock, in the sense of exercising a majority
influence upon some later $2D+1$-VSRC, no such lock can safely be deleted
from $\hist_i$ of any $p_i$ after bounded time.

\end{compactitem}

   \subsection{Correctness Proof}\label{sec:corproof}
   
In this final subsection, we will prove the following \cref{thm:kset}:

\begin{theorem} \label{thm:kset}
\cref{alg:ksa} solves $k$-uniform $k$-set agreement in a dynamic network under
the message adversary \MAkI, which is the conjunction of \cref{ass:inter}
and \cref{ass:stablestream}.
\end{theorem}

The proof consists of a sequence of technical lemmas, which will
finally allow us to establish all the properties of $k$-set
agreement given in \cref{sec:model}. First,
validity according to \cref{def:ksa} is straightforward to see, as only the values of locks
are ever considered as decisions (\cref{line:setDecision}).
Values of locks, on the other hand, are initialized to the initial value of
a process (\cref{line:initHist}) and later on always have values of previous
locks assigned to them
(\cref{line:createLockMajority,line:createLockMaximum}).
Note that the claimed
$k$-uniformity is obvious, as the code of the algorithm does
not involve $k$.

To establish termination, we start with some simple properties
related to setting locks at all members of vertex stable root
components.

\begin{lemma}\label{lem:minlock}
Apart from processes adopting a decision sent by another process,
only processes part of a vertex stable root with interval length greater than
$D$ (resp.\ $2D$) lock (resp.\ decide).
\end{lemma}
\begin{proof}
The if-statement in \cref{line:lock?} (resp. \cref{line:decide?}) is
evaluated to true only if $\stableSCC$ detects a stable member set $R$
in some interval $I$ of length $D+1$ (resp.\ of length $2D+1$) or larger, which implies by
\cref{cor:stable2root} that $R^I$ is indeed a $D+1$-VSRC (resp.\ $2D+1$-VSRC). \qed
\end{proof}

\begin{lemma}\label{lem:maxlock}
All processes part of a vertex stable root $R^{[a,b]}$
with interval length greater than $2D$, which did not start
already before $a$, lock, i.e. set $\ell:=a$, in round $a+2D$.
\end{lemma}
\begin{proof}
Because $R^{[a,b]}$ is $D$-bounded by \cref{ass:inter}, \cref{cor:allrootdec} guarantees
that $\stableSCC(a,a+D)$ returns $R$ from round $a+2D$ (of the $k$-set-algorithm)
on, and that it cannot have done so already in round $a+2D-1$.
Hence, $\ell=\bot$ in round $a+2D$,
the if-statement in \cref{line:lock?}
is entered
and $\ell:=a$ is set in \cref{line:setLock}. \qed
\end{proof}

\begin{lemma}\label{lem:maxdecide}
All processes part of a vertex stable root  $R^{[a,b]}$
with interval length greater than $3D$, which did not start
already before $a$, have decided by round $a+3D$.
\end{lemma}
\begin{proof}
It follows from \cref{lem:maxlock} that all members of the VSRC $R^{[a,b]}$
set $\ell:=a$ in round $a+2D$.
As the VSRC remains stable also in rounds
$a+2D,\dots,a+3D$, \cref{line:releaseLock} will not be executed in
these rounds, thus $\ell=a$ remains unchanged.
Consequently, due to \cref{cor:allrootdec}, the if-statement
in \cref{line:decide?} will evaluate to true at the latest in round
$\ell+3D=a+3D$, causing all the processes to decide via
\cref{line:setDecision} by round $a+3D$ as asserted. \qed
\end{proof}

\begin{lemma}\label{lem:termination}
The algorithm eventually terminates at all processes.
\end{lemma}
\begin{proof}
For a contradiction, assume that there is $p_j \in \Pi$ which has not terminated
after the stable interval guaranteed by \cref{ass:inter}.
This implies that $p_j$ is not part of a root component during this stable interval,
because \cref{lem:maxdecide} ensures termination by $r_{ST}+3D$ at the latest
for the latter. Hence, $p_j$ did not get a decide message either. From
\cref{def:E-network-vertex-stable-roots}, it follows
that there exists a causal chain of length at most $\nwbound$ to $p_j$ from some
member $p_i$ of a VSRC after its termination.
Therefore, it must receive the decide message by
$r_{ST}+3D+\nwbound$
at latest. \qed
\end{proof}

Although we now know that all members of a VSRC that is vertex
stable for at least $3D$ rounds will decide, we did not prove anything
about their decision values yet. In the sequel, we will prove
that they decide on the \emph{same} value.

\begin{lemma} \label{lem:D-consistent-history}
 Given some VSRC $R^I$ with $I=[a,b]$ and $b \ge a+D$, in all rounds
 $x \in [a+D, b]$ it holds that $\forall p_i, p_j \in R \colon
 \bigcup_{r' \leq a} \hist_i[j][r'] = \bigcup_{r' \leq a} \hist_j[j][r']$
\end{lemma}

\begin{proof}
 By the $D$-boundedness of $R^I$, a message from round $a$ has reached
 every member of $R$ by round $a+D$.
 Moreover, no message sent by a process not in $R$ during $I$ can
 reach a member of $R$ during $I$ because $R^I$ is a root component.
 Therefore,
 since $\hist_i$ is sent by each process $p_i$ in every round
 (\cref{line:sendHistory})
 and $p_i$ adds only newly learned entries to $\hist_i$
 (\cref{line:addNewLock,line:learnLock}),
 all these updates of $\hist_i$
 during $I$, regarding any round $r'\leq a$, occur at the
 latest in round $a+D$. \qed
\end{proof}

\begin{lemma} \label{lem:consistent-decisions}
 All processes of a VSRCs $R^I$ of $\mathbb{V}_{2D+1}$ with $I=[a,b]$
 adopt the same lock (and hence decide the same).
\end{lemma}

\begin{proof}
Such a lock is created by $p_i \in R$ in round $a+2D$,
when it recognizes $R^I$ as having been vertex-stable for $D+1$ rounds
according to \cref{lem:maxlock}. As the lock (value) is computed based on $\hist_i$
present in round $a+2D$, which is consistent among all VSRC members
by \cref{lem:D-consistent-history}, the lemma follows. \qed
\end{proof}

Finally, we show that, given that the system satisfies \cref{ass:stablestream},
there will be at most $k$ decision values in any run of \cref{alg:ksa}, which
proves $k$-agreement:
Since there are at most $k$ VSRCs of $\mathbb{V}_{2D+1}$
that are not majority-influenced by other VSRCs,
it remains to show that any majority-influenced VSRC decides the same as the
VSRC it is majority-influenced by.
In order to do so, we will first establish a key property of our central
data structure $\hist_i$.

\begin{lemma} \label{lem:historyApproximation}
Given $R_{cur}^{I_{cur}=[r_{cur},s_{cur}]}$,
       $R_{suc}^{I_{suc}=[r_{suc},s_{suc}]}$
with $|I_{cur}| > 2D$ and any $|I_{suc}| \geq 1$.
Let $L$ be a lock known to all members of $R_{cur}$ by $s_{cur}$,
i.e., for all $p_i \in R_{cur}$ it holds that, by the end of round $s_{cur}$,
$L \in \bigcup_{r' \leq s_{cur}} \hist_i[i][r']$.
For any process $p_j \in R_{suc}$, it holds that
if there exists some $p_i \in R_{cur}$, s.t.\ $\influence{p_i}{p_j}$,
then
$L \in \bigcup_{r' \leq r_{suc}} \hist_j[j][r']$.

\end{lemma}
\begin{proof}
    Assume there exists a $p_i \in R_{cur}$ s.t.\ $\influence{p_i}{p_j}$ but
    $L \notin \bigcup_{r' \leq r_{suc}} \hist_j[j][r']$.
    The definition of
    $\influence{p_i}{p_j}$ implies that
    there exists a causal chain from $p_i$ to $p_j$ that ends before
    $p_j$ becomes a part of $R_{suc}$.
    Since processes send their own history in every round according to
    \cref{line:sendHistory}, every message in this causal chain consisted of
    a $\hist$ containing $L$ and thus
    $p_j$ put $L$ into its $\hist_j[j][r]$ via
    \cref{line:remoteHistoryUpdate} if $\bigcup_{r' \leq r} \hist_j[j][r']$
    did not already contain $L$. \qed
\end{proof}

\begin{lemma}
 Given $R_{cur}^{I_{cur}=[r_{cur},s_{cur}]}\in \mathbb{V}_{2D+1}$ and
       $R_{suc}^{I_{suc}=[r_{suc},s_{suc}]} \in \mathbb{V}_{2D+1}$,
 assume that the processes of $R_{cur}$ created the (same) lock $L$ when locking.
 If $\majinf{R_{cur}^{I_{cur}}}{R_{suc}^{I_{suc}}}$, then the processes of $R_{suc}$ will choose
 a lock $L'$ where $L.v = L'.v$ (and hence decide the same as the processes of
 $R_{cur}$).
\end{lemma}

\begin{proof}
 From the definition of $\majinf{}{}$ (\cref{def:majorityInfluence}), it
 follows that no VSRC $R^I$ of $\mathbb{V}_{D+1}$ has a larger influence set on
 $R_{suc}$
 than $R_{cur}$.
 By \cref{lem:minlock}, this implies that no lock that was generated by some
 $R^I$ in $\mathbb{V}_{D+1}$ can be known to more members of $R_{suc}$ than the
 lock $L$
 generated by $R_{cur}$.
 Since process $p_i$ puts only newly learned locks into $\hist_i$
 (\cref{line:addNewLock,line:learnLock}),
 by \cref{lem:historyApproximation}, this means that in round $r_{suc}$
 no ``bad'' lock $L_b$ is present in more elements of
 $S = \bigcup_{p_i \in R_{suc},r' \leq r_{suc}} \hist_i[i][r']$
 than $L$.
 We now show that $L.\hot > L_b.\hot$ for all $L_b$ occuring in as many
 elements of $S$ as $L$ with $L_b \neq L$.
 Obviously, the only locks $L_b$ that could occur in as many elements of $S$
 as $L$ are locks that have been in $\hist_i$ of some $p_i \in R_{cur}$
 at the beginning of round $r_{cur}$ already.
 Since for any such $L_b$, $L$ was created after $L_b$,
 by \cref{line:createLockMajority,line:createLockMaximum},
 we have that $L.\hot > L_b.\hot$, as claimed.
 Because in round $r_{suc}+2D$, at all processes $p_i, p_j$ of $R_{suc}$,
 \cref{lem:D-consistent-history} implies
 that $\bigcup_{r' \leq r_{suc}} \hist_i[j][r']
     = \bigcup_{r' \leq r_{suc}} \hist_j[i][r']$,
 when locking in round $r_{suc}+2D$ according to \cref{lem:maxlock},
 every $p_i$ of $R_{suc}$ will find $L$ as the unique most common lock in the
 elements of $S$ with maximal $\hot$.
 This leads to the evaluation of the if-statement in \cref{line:hottestMajority?}
 to true and to the creation of a new lock $L'$,
 where $L'.v = L.v$ in \cref{line:createLockMajority}, as asserted. \qed
\end{proof}

This finally completes the proof of \cref{thm:kset}.

\section{Conclusions}
\label{sec:conclusion}

We introduced a framework for modeling dynamic networks with directed
communication links under generalized message adversaries that focus
on vertex-stable root components. We presented
related impossibility results and lower bounds for consensus, as well
as a message adversary that is much stronger than the ones known so far
for solving consensus, along with a suitable algorithm and its correctness
proof. Moreover, we made a significant step towards determining the 
solvability/impossibility border of general $k$-set agreement in our model. 
We provided several impossibility results and lower
bounds, which also led us to the, to the best of our knowledge, first gracefully degrading 
consensus/$k$-uniform $k$-set agreement under fairly strong message adversaries
proposed so far.

\iftoggle{journal}
{
\bibliography{lit,additional}
}
{
\bibliography{paper}
}
\bibliographystyle{abbrv}

\iftoggle{journal}{}{
\clearpage
\appendix
\section{Appendix}

\subsection{Generic Impossibility Theorem}

For the convenience of the reviewers, we restate the generic impossibility 
theorem from \cite{BRS11:OPODIS}.

\begin{theorem}[$k$-Set Agreement Impossibility {\cite[Thm.~1]{BRS11:OPODIS}}]
\label{thm:impossibility}
Let $\Mod=\li{\Pi}$ be a system model and consider the runs $\M_A$ that are
generated
  by some fixed algorithm $A$ in $\Mod$, where every process starts
  with a distinct input value.
Fix some nonempty and pairwise disjoint sets of processes
$\D_1,\dots,\D_{k-1}$, and a
      set of distinct decision values $\set{v_1,\dots,v_{k-1}}$.
Moreover, let $\D=\bigcup_{1\le i< k}\D_i$ and $\PiD=\Pi\setminus \D$.
Consider the following two properties:
\begin{itemize}
\item[\boldmath\bf(dec-$\D$)]
For every set $\D_i$, value $v_i$ was proposed by
  some $p\in \D$, and there is some $q\in \D_i$ that decides $v_i$.
\item[\boldmath\bf(dec-$\PiD$)]
If $p_j \in \PiD$ then $p_j$ receives
  no messages from any process in $\D$ until every process in $\PiD$ has decided.
\end{itemize}
Let $\Rii\subseteq \M_A$ and $\Riii\subseteq \M_A$ be the sets
  of runs of $A$ where (dec-$\PiD$) respectively both, (dec-$\D$) and
(dec-$\PiD$), hold.\footnote{Note that
  $\Rii$ is by definition compatible with the runs of the restricted
algorithm $\APiD$.}
Suppose that the following conditions are satisfied:
\begin{description}
\item[\bf\textup\condNonempty]
$\Rii$ is nonempty.
\item[\bf\textup\condRiiRiii]
$\Rii\subruns[\PiD]\Riii$.
\end{description}
In addition, consider a restricted model $\Mod'=\li{\PiD}$ such that the
following properties hold:
\begin{description}
\item[\bf\textup\condNocons]
There is no algorithm that solves consensus in $\Mod'$.
\item[\bf\textup\condMADMA]
$\MAD'\subruns \M_A$.
\end{description}
Then, $A$ does not solve $k$-set agreement in $\Mod$.
\end{theorem}

}
\end{document}